\documentclass[aps,pra,reprint,amsmath,amssymb,superscriptaddress,floatfix]{revtex4-2}
\usepackage{bm}
\usepackage{graphicx}
\usepackage{booktabs}
\usepackage{mathtools}
\usepackage{microtype}
\usepackage{amsthm}
\usepackage[hidelinks]{hyperref}
\newtheorem{proposition}{Proposition}
\newtheorem{theorem}{Theorem}
\newtheorem{lemma}{Lemma}
\newtheorem{remark}{Remark}
\newcommand{\Tr}{\operatorname{tr}}
\newcommand{\Exc}{\mathcal E}
\graphicspath{{figures/}}
\newcommand{\mv}{\boldsymbol m}
\newcommand{\nop}{\hat n}

\begin{document}

\title{Exact damping-basis continued fractions for few-atom lasers:\\
collective correlations beyond second-order cumulant closure}

\author{Stephan Hartmann}
\affiliation{Munich Center for Mathematical Philosophy, LMU Munich, Germany}

\date{\today}

\begin{abstract}
We combine the Briegel--Englert damping basis of a lossy cavity mode with
the permutation-invariant Liouville basis of \(N\) identical, incoherently
pumped two-level emitters.  The resulting Liouvillian is exactly block
tridiagonal in a single radial field index, with atomic blocks of
dimension \(\binom{N+3}{3}\).  The stationary state follows from a matrix
continued fraction whose asymptotic closure is a linear matrix equation
with a unique, nilpotent solution throughout the dissipative regime
\(A>0\).  For a zero-temperature cavity we prove terminal-condition
independence of the backward recursion and identify its limit with the
stationary-state coefficients under an explicit moment-stationarity
uniqueness hypothesis and an invertibility condition on the limiting
pivots; the same construction in the coherence-one sector yields the
emission spectrum.  The method is validated against direct
diagonalization, weak-pump perturbation theory, and closed-system spectra,
reaches \(N=16\) in tens of seconds in the few-quanta regime, and has
quantified precision limits at large photon number and at large atom
number.  We use it to derive an
exact two-atom photon-balance identity that isolates the cavity-induced
pair coherence and relates it to a singlet--triplet population difference;
to prove pump-order upper bounds for excitation cumulants and connected
inter-emitter coherences at zero cavity temperature, whose leading
coefficients, which the bounds do not fix, we evaluate numerically
through eighth order without radial truncation error; to establish a
parity property of the minimal collective-spin sector; and to benchmark
the second-order cumulant closure
used in few-emitter laser theory, which we find to be semi-quantitative
for the photon number but to predict the wrong sign of the inter-emitter
coherence at strong pump.  In a good-cavity regime the exact solution
shows near-Poissonian statistics and strong line narrowing for two and
three emitters.
\end{abstract}

\maketitle

\section{Introduction}
\label{sec:intro}

Microscopic lasers test quantum-optical laser theory at its foundations,
because the gain medium can no longer be replaced by a macroscopic
semiclassical polarization.  The one-atom laser and micromaser literature
established that lasing-like behavior, nonclassical photon statistics, and
strong atom--field correlations already occur at the level of a few quanta
\cite{Filipowicz1986,MuSavage1992,An1994,Loffler1997,Kilin2002,McKeever2003},
and few-emitter lasing has since been reached in ultrasmall plasmonic
nanocavities \cite{Ojambati2024}, in steady-state superradiant lasers
with less than one intracavity photon \cite{MeiserHolland2010,Bohnet2012},
and in quantum-dot nanolasers with giant photon bunching
\cite{Jahnke2016}.  When the number of emitters is small, the theoretical
question is not only how to compute, but which approximations remain
trustworthy: the workhorse of few-emitter laser theory is the cluster or
cumulant expansion truncated at second order
\cite{Gies2007,Leymann2014,MeiserHolland2010,KirtonKeeling2018,Ojambati2024},
whose accuracy for a handful of strongly coupled emitters cannot be judged
from within the method.  Exact finite-\(N\) solutions that retain the
collective atomic structure and the full cavity Hilbert space are the
natural benchmark.

Three lines of prior work each supply part of such a solution.  Briegel
and Englert diagonalized the dissipative field and atomic Liouvillians in
\emph{damping bases} before treating their coupling
\cite{BriegelEnglert1993,EnglertMorigi2002}, and Ginzel \emph{et al.}
solved the continuously pumped one-atom laser in this framework
\cite{Ginzel1993}; the construction is exact in the field but was applied
to a single emitter.  Permutation symmetry reduces the Liouville space of
\(N\) identical two-level systems with symmetric local dissipation from
\(4^N\) to \(D_N=\binom{N+3}{3}\) dimensions
\cite{ChaseGeremia2008,Baragiola2010,XuTieriHolland2013,Hartmann2016,GeggRichter2016,GeggRichter2017,Shammah2018,KirtonKeeling2017,KirtonKeeling2018};
these methods handle many emitters exactly but truncate the cavity in the
Fock basis and solve a finite sparse Liouvillian.  Matrix continued
fractions, finally, are the classical tool for tridiagonal hierarchies in
Fokker--Planck and quantum-optical master equations
\cite{Risken1989,RiskenSavageHaakeWalls1987,VogelRisken1988} and continue
to be used for block-tridiagonal Liouvillians \cite{KeliriSchiro2026}.
Two-atom lasers were studied directly by Yeoman and Meyer, who attributed
an improved laser line to a cavity-induced mutual atomic coherence
\cite{YeomanMeyer1998,MeyerYeoman1997,Yeoman1998}, and by others
\cite{GheriHorakRitsch1997,KolobovHaake1997}; Fleischhauer's
\(N\)-atom/one-atom relation \cite{Fleischhauer1994} and the few-emitter
cooperativity study of Auff\`eves \emph{et al.} \cite{Auffeves2011}
provide limiting benchmarks.

The contribution of this paper is the combination of the first two
ingredients, which exposes a structure neither has alone.  Because the
Briegel--Englert multiplication algebra changes the radial field index by
at most one and the permutation-invariant atomic basis is closed under
collective raising and lowering, the coupled Liouvillian of the pumped
Tavis--Cummings model is an exact block-Jacobi operator in one infinite
index, with \(D_N\times D_N\) blocks (Proposition~\ref{prop:blockjacobi}).
The stationary state is then a matrix continued fraction whose asymptotic
closure is a linear matrix equation, uniquely solvable for every parameter
set with \(A>0\) (Proposition~\ref{prop:linear-closure}), and the emission spectrum
follows from the same recurrence in the coherence-one sector.  The
closure, the off-diagonal blocks, and the asymptotic transfer matrix are
all nilpotent with respect to a natural grading
(Proposition~\ref{prop:grading}); this makes the backward recursion
forget its terminal condition super-exponentially
(Proposition~\ref{prop:terminal}) and, combined with a uniform
exponential moment bound for the stationary state
(Lemma~\ref{lem:expmoments}), lets us identify, for a zero-temperature cavity and under an explicit
moment-stationarity uniqueness hypothesis (U) and an invertibility
condition (P) on the limiting pivots, the limit of the recursion with
the coefficients of the stationary state of the infinite-dimensional
master equation (Theorem~\ref{thm:branch}), not merely with a sequence
of finite truncations.  The
permutation-invariant atomic sector supplies exact collective
correlations; the damping basis supplies analytic access to observables
and to weak-pump structure.  The price is a precision limit at large
photon number and, through the conditioning of the recursion, at large
atom number, both of which we quantify.

We use the method for three purposes.  First, we derive exact identities
for collective correlations in two and three emitters: a photon-balance
identity in which the cavity-induced pair coherence appears as a separate
term equal to a singlet--triplet population difference; a proved
pump-order hierarchy (upper bounds) for excitation cumulants and
connected coherences, whose leading coefficients we evaluate through
eighth order in the pump from finite linear systems free of radial
truncation error; and a parity property of
the minimal collective-spin sector that is sharp in sector-resolved
quantities and invisible in the photon statistics.  Second, we benchmark
the second-order cumulant closure against the exact solution for the same
master equation, up to \(N=16\), and find that it is semi-quantitative
for the photon number but fails qualitatively for the inter-emitter
coherence at strong pump.  Third, we follow the same model from the
bad-cavity few-quanta regime, whose spectrum is a collective Rabi
doublet, into a good-cavity regime of genuine few-emitter lasing with
near-Poissonian statistics and line narrowing, where the two-atom line
improvement of Yeoman and Meyer is reproduced by the full calculation.

Section~\ref{sec:model} fixes the model and the two bases,
Sec.~\ref{sec:reduction} states the block-Jacobi reduction, and
Sec.~\ref{sec:cf} develops the continued fraction, its convergence and
branch selection, its conditioning, the
observable functionals, and the spectrum.  Sections~\ref{sec:collective},
\ref{sec:benchmark}, and \ref{sec:lasing} contain the three applications,
Sec.~\ref{sec:discussion} the comparison with other methods and the
limitations, and the appendices the field coefficients, observable
functionals, the proof of Theorem~\ref{thm:branch}, validation tests, and
maximum-entropy details.  All code and
data are available \cite{code}.

\section{Model and permutation-invariant damping basis}
\label{sec:model}

We consider \(N\) identical two-level atoms coupled symmetrically to one
lossy cavity mode, with \(\tau_+=\lvert e\rangle\langle g\rvert\),
\(\tau_-=\tau_+^\dagger\), \(\sigma_z=\lvert e\rangle\langle e\rvert-\lvert g\rangle\langle g\rvert\),
collective operators \(S_\pm=\sum_j\tau_\pm^{(j)}\), \(S_z=\sum_j\sigma_z^{(j)}\),
photon-number operator \(\nop=a^\dagger a\), and Hamiltonian
\begin{equation}
H=\hbar\omega\nop+\frac{\hbar\Omega}{2}S_z-\frac{\hbar g}{2}(a^\dagger S_-+aS_+),\qquad \Delta=\Omega-\omega .
\label{eq:H}
\end{equation}
With \(\mathcal D[L]\rho=L\rho L^\dagger-\tfrac12\{L^\dagger L,\rho\}\), the
master equation \cite{Lindblad1976,GKS1976} is
\begin{align}
\dot\rho={}&-\frac{i}{\hbar}[H,\rho]+A(1+\nu)\mathcal D[a]\rho
+A\nu\mathcal D[a^\dagger]\rho\notag\\
&+\sum_{j=1}^N\Big[B(1-s)\mathcal D[\tau_-^{(j)}]\rho
+Bs\mathcal D[\tau_+^{(j)}]\rho\notag\\
&\hspace{14mm}+\frac{C-B/2}{2}\mathcal D[\sigma_z^{(j)}]\rho\Big].
\label{eq:master}
\end{align}
Here \(A\) is the cavity energy-decay rate, \(\nu\) its thermal occupation,
\(B\) the longitudinal atomic rate, \(s\in[0,1]\) the pump parameter
(\(Bs\) is the incoherent pump rate, \(B(1-s)\) the decay rate, and \(s\)
the uncoupled excited-state population), and \(C\ge B/2\) the transverse
rate, with \(C=B/2\) meaning no pure dephasing.  Briegel and Englert use
\(\sigma_\pm^{\rm BE}=2\tau_\pm\), so \(g_{\rm BE}=g/2\).  Throughout we
take \(A,B>0\), \(g\neq0\), and \(0<s<1\), so that cavity loss, the
coupling, and the local atomic pump and decay channels are all present,
and we do not consider boundary cases such as \(s=1\), \(B=0\), or
\(g=0\) (at \(g=0\) the cavity decouples and the truncated dynamics is
no longer irreducible, although its stationary state is still unique).  For every finite Fock truncation of
the cavity (cutoff \(n_{\rm c}\), Hamiltonian replaced by its compression
to the truncated space) the resulting Lindbladian is irreducible ---
local pump and decay connect all atomic configurations, and the
Tavis--Cummings coupling together with cavity loss connects all Fock
levels --- and hence has a unique stationary state \(\rho^{(n_{\rm c})}\)
by the criteria of Refs.~\cite{Spohn1977,Frigerio1978}.  For the
untruncated equation at zero cavity temperature, \(\nu=0\) --- the case
of all theorems below --- we work under the following standing
hypothesis,
with \(\Exc=\nop+\sum_j\tau_+^{(j)}\tau_-^{(j)}\) the total excitation
operator and \(\mathcal T_{\exp}\) the set of trace-class operators
\(\sigma\) for which \(e^{\theta\Exc}\sigma e^{\theta\Exc}\) is trace
class for every \(\theta>0\) (a class tailored to \(\nu=0\), where the
stationary photon distribution has super-exponential tails,
Lemma~\ref{lem:expmoments}; for a thermal cavity it is not the natural
class):
\begin{itemize}
\item[(U)] Every \(\sigma\in\mathcal T_{\exp}\) with
\(\Tr(\sigma\,\mathcal L^\dagger Y)=0\) for all polynomials \(Y\) in
\(a,a^\dagger\) and the atomic operators is a multiple of one density
operator \(\rho_{\rm ss}\).
\end{itemize}
Section~\ref{sec:branch} shows that under (U) the truncated states
\(\rho^{(n_{\rm c})}\) converge to \(\rho_{\rm ss}\) and that, under (U)
and a pivot condition (P) stated there, the continued fraction computes
\(\rho_{\rm ss}\) itself.  Hypothesis (U) is the
infinite-dimensional counterpart of the finite-truncation uniqueness.
It is expected to hold: the generator admits the Lyapunov function
\(\Exc\), with \(\mathcal L^\dagger\Exc=-A\nop-BN_e+BsN+A\nu\le
-\min(A,B)\,\Exc+BsN+A\nu\) and \(\Exc\) of compact resolvent, which is
the type of condition under which the minimal quantum dynamical semigroup
is conservative \cite{ChebotarevFagnola1998} and possesses a stationary
state \cite{FagnolaRebolledo2001}, while irreducibility of the kind just
described yields its uniqueness
\cite{FagnolaRebolledo2002,Frigerio1978}, and uniqueness of the invariant
state implies that the kernel of the generator on trace-class operators is
one-dimensional (the positive and negative parts of a stationary
self-adjoint trace-class operator are separately stationary, because a
trace-preserving positive evolution cannot increase the trace norm and
the Jordan decomposition is the unique norm-additive one).  Hypothesis
(U) asks for slightly more than this: it is phrased in terms of the
polynomial moment equations rather than of the generator of the
semigroup, and so presupposes that an element of \(\mathcal T_{\exp}\)
that is stationary in the moment sense is invariant under the semigroup
--- a core property of the polynomial observables that the cited
results do not supply directly.  We have not verified these domain
conditions for the present model and therefore state (U) as a
hypothesis rather than a theorem.

The cavity Liouvillian
\(\mathcal L_{\rm f}=A(1+\nu)\mathcal D[a]+A\nu\mathcal D[a^\dagger]-i\omega[\nop,\cdot]\)
has right eigenoperators \(\rho_n^{(k)}\), \(n\ge0\), \(k\in\mathbb Z\),
\begin{equation}
\mathcal L_{\rm f}\rho_n^{(k)}=\lambda_{n,k}^{\rm f}\rho_n^{(k)},\qquad
\lambda_{n,k}^{\rm f}=-A\Big(n+\frac{|k|}{2}\Big)-i\omega k ,
\label{eq:fieldeigen}
\end{equation}
where \(n\) is the radial damping index and \(k\) the coherence order or
\emph{grade} \cite{BriegelEnglert1993}.  Left or right multiplication by
\(a\) or \(a^\dagger\) changes \(k\) by one and \(n\) by at most one; the
twelve identities are listed in Appendix~\ref{app:fieldcoeff}, and
\(\operatorname{tr}\rho_n^{(k)}=\delta_{n0}\delta_{k0}\).

For one atom the damping eigenoperators of the local dissipator in
Eq.~(\ref{eq:master}) are \(r_0=(1-s)|g\rangle\langle g|+s|e\rangle\langle e|\),
\(r_z=\sigma_z\), \(r_\pm=\tau_\pm\), with eigenvalues \(0,-B,-C,-C\).  A
permutation-invariant \(N\)-atom operator is labelled by the composition
\(\mv=(m_0,m_z,m_+,m_-)\), \(m_0+m_z+m_++m_-=N\), and defined as
\begin{align}
R_{\mv}&=\frac{1}{\mathcal N(\mv)}
\sum_{\pi}^{\rm distinct}\pi\big[r_0^{\otimes m_0}\otimes r_z^{\otimes m_z}
\otimes r_+^{\otimes m_+}\otimes r_-^{\otimes m_-}\big],\notag\\
\mathcal N(\mv)&=\frac{N!}{m_0!\,m_z!\,m_+!\,m_-!},
\label{eq:Rn-explicit}
\end{align}
the \emph{average} over the \(\mathcal N(\mv)\) distinct arrangements.
This is the generalized-Dicke Liouville basis of Ref.~\cite{Hartmann2016}
written in the one-atom damping eigenbasis; it is equivalent to the
\(SU(4)\) construction of Ref.~\cite{XuTieriHolland2013} and to the
symmetric bases of Refs.~\cite{ChaseGeremia2008,GeggRichter2016,Shammah2018}.
Its dimension is
\begin{equation}
D_N=\binom{N+3}{3},\quad D_1=4,\ D_2=10,\ D_3=20,\ D_{16}=969 .
\label{eq:DN}
\end{equation}
Permutation invariance of the density operator should not be confused
with restriction to the maximal-spin Hilbert subspace: the invariant
Liouville space contains the populations and coherences of \emph{all}
total-spin sectors, including the dark and subradiant ones discussed in
Sec.~\ref{sec:sectors}.  Each \(R_{\mv}\) is an eigenoperator of the
uncoupled atomic Liouvillian with eigenvalue
\(-Bm_z-C(m_++m_-)-i\Omega q(\mv)\), where \(q(\mv)=m_+-m_-\) is the
atomic coherence charge, and \(\operatorname{tr}R_{\mv}=\delta_{\mv,(N,0,0,0)}\).
Because \(S_\pm\) are themselves \(S_N\)-invariant, left and right
multiplication by \(S_\pm\) closes on this basis:
\begin{align}
S_+R_{\mv}={}&(1-s)m_0R_{m_0-1,m_z,m_++1,m_-}\notag\\&-m_zR_{m_0,m_z-1,m_++1,m_-}\notag\\
&+m_-R_{m_0+1,m_z,m_+,m_--1}\notag\\&+(1-s)m_-R_{m_0,m_z+1,m_+,m_--1},\label{eq:SLplus}\\
S_-R_{\mv}={}&sm_0R_{m_0-1,m_z,m_+,m_-+1}\notag\\&+m_zR_{m_0,m_z-1,m_+,m_-+1}\notag\\
&+m_+R_{m_0+1,m_z,m_+-1,m_-}\notag\\&-sm_+R_{m_0,m_z+1,m_+-1,m_-},\label{eq:SLminus}\\
R_{\mv}S_+={}&sm_0R_{m_0-1,m_z,m_++1,m_-}\notag\\&+m_zR_{m_0,m_z-1,m_++1,m_-}\notag\\
&+m_-R_{m_0+1,m_z,m_+,m_--1}\notag\\&-sm_-R_{m_0,m_z+1,m_+,m_--1},\label{eq:SRplus}\\
R_{\mv}S_-={}&(1-s)m_0R_{m_0-1,m_z,m_+,m_-+1}\notag\\&-m_zR_{m_0,m_z-1,m_+,m_-+1}\notag\\
&+m_+R_{m_0+1,m_z,m_+-1,m_-}\notag\\&+(1-s)m_+R_{m_0,m_z+1,m_+-1,m_-},\label{eq:SRminus}
\end{align}
where terms with a negative occupation are absent.  These follow from
\(\tau_+r_0=(1-s)\tau_+\), \(\tau_+\sigma_z=-\tau_+\),
\(\tau_+\tau_-=r_0+(1-s)r_z\), \(\tau_-\tau_+=r_0-sr_z\), their
right-multiplication counterparts, and the combinatorial factor
\(\mathcal N(\mv')/\mathcal N(\mv)\) produced by the averaging in
Eq.~(\ref{eq:Rn-explicit}).  They are checked coefficient by coefficient
against explicit tensor products for \(N\le5\) in the test suite of
Ref.~\cite{code} (Appendix~\ref{app:validation}).

\section{Exact block-Jacobi reduction}
\label{sec:reduction}

A combined basis element is \(\rho_n^{(k)}\otimes R_{\mv}\), and the
total coherence
\begin{equation}
K=k+q(\mv)
\label{eq:K}
\end{equation}
is conserved by the interaction, which pairs \(a^{(\dagger)}\) with
\(S_\mp\).  Fix an ordering \(\{\mv_\alpha\}_{\alpha=1}^{D_N}\), set
\(q_\alpha=q(\mv_\alpha)\), \(k_\alpha=K-q_\alpha\), and
\(\Phi_{n\alpha}^{(K)}=\rho_n^{(k_\alpha)}\otimes R_{\mv_\alpha}\); any
operator in the fixed-\(K\) sector is
\(\rho^{(K)}=\sum_{n\ge0}\sum_\alpha(X_n)_\alpha\Phi_{n\alpha}^{(K)}\) with
\(X_n\in\mathbb C^{D_N}\).

Let \(A_\mu^{\rm L},A_\mu^{\rm R}\) (\(\mu=\pm\)) be the matrices of
Eqs.~(\ref{eq:SLplus})--(\ref{eq:SRminus}),
\(S_\mu R_{\mv_\beta}=\sum_\alpha(A_\mu^{\rm L})_{\alpha\beta}R_{\mv_\alpha}\),
\(R_{\mv_\beta}S_\mu=\sum_\alpha(A_\mu^{\rm R})_{\alpha\beta}R_{\mv_\alpha}\),
and \(f^{\rm L/R}_{\mu,\delta}(n,k)\) the field coefficients
\begin{align}
a_\mu\rho_n^{(k)}&=\sum_{\delta=-1}^{1}f_{\mu,\delta}^{\rm L}(n,k)\rho_{n+\delta}^{(k+\mu)},\notag\\
\rho_n^{(k)}a_\mu&=\sum_{\delta=-1}^{1}f_{\mu,\delta}^{\rm R}(n,k)\rho_{n+\delta}^{(k+\mu)},
\label{eq:field-action-coeff}
\end{align}
with \(a_+=a^\dagger\), \(a_-=a\) (Appendix~\ref{app:fieldcoeff}).  The
uncoupled part is diagonal,
\begin{align}
d_{n\beta}^{(K)}=-A\Big(n+\frac{|k_\beta|}{2}\Big)-Bm_{z,\beta}-C(m_{+,\beta}+m_{-,\beta})\notag\\
\hspace{18mm}-i\omega K-i\Delta q_\beta ,
\label{eq:uncoupled-diagonal}
\end{align}
where \(-i\omega k_\beta-i\Omega q_\beta=-i\omega K-i\Delta q_\beta\): in
the stationary sector \(K=0\) the frame frequency drops out, and for
\(K\neq0\) it is a uniform shift.  The interaction
\(\mathcal L_{\rm int}\rho=\tfrac{ig}{2}[a^\dagger S_-+aS_+,\rho]\)
contributes the source-to-target radial blocks
\begin{align}
(\mathcal V_n^{(\delta)})_{\alpha\beta}=\frac{ig}{2}\Big[&
 f_{+,\delta}^{\rm L}(n,k_\beta)(A_-^{\rm L})_{\alpha\beta}
+f_{-,\delta}^{\rm L}(n,k_\beta)(A_+^{\rm L})_{\alpha\beta}\notag\\
&-f_{+,\delta}^{\rm R}(n,k_\beta)(A_-^{\rm R})_{\alpha\beta}
-f_{-,\delta}^{\rm R}(n,k_\beta)(A_+^{\rm R})_{\alpha\beta}\Big],
\label{eq:Vdelta}
\end{align}
and with \(D_n=\operatorname{diag}(d_{n\beta}^{(K)})\),
\begin{equation}
M_n=D_n+\mathcal V_n^{(0)},\quad G_n=\mathcal V_{n+1}^{(-1)},\quad
F_n=\begin{cases}0,&n=0,\\ \mathcal V_{n-1}^{(+1)},&n\ge1,\end{cases}
\label{eq:MGF}
\end{equation}
\(G_n\) collecting transitions from source level \(n+1\) down to \(n\) and
\(F_n\) from \(n-1\) up to \(n\).

\begin{proposition}[Exact block-Jacobi reduction]
\label{prop:blockjacobi}
For every finite \(N\) and fixed total coherence \(K\), the Liouvillian of
Eq.~(\ref{eq:master}) restricted to the permutation-invariant sector
spanned by \(\{\Phi_{n\alpha}^{(K)}\}\) is exactly block tridiagonal in the
radial index,
\begin{equation}
\dot X_n=M_nX_n+G_nX_{n+1}+F_nX_{n-1},\qquad X_{-1}=0,
\label{eq:blockjacobi}
\end{equation}
with \(M_n,G_n,F_n\in\mathbb C^{D_N\times D_N}\) given by Eq.~(\ref{eq:MGF}).
\end{proposition}
\begin{proof}
The uncoupled Liouvillian is diagonal in the product basis by
Eq.~(\ref{eq:fieldeigen}) and the eigenoperator property of \(R_{\mv}\);
the interaction conserves \(K\) because it raises \(k\) by one exactly
when it lowers \(q\) by one and vice versa; and the twelve field
identities of Appendix~\ref{app:fieldcoeff} change \(n\) by
\(\delta\in\{-1,0,1\}\) only.  Hence the coefficient of
\(\Phi_{n\alpha}^{(K)}\) in \(\mathcal L\Phi_{m\beta}^{(K)}\) vanishes unless
\(n\in\{m-1,m,m+1\}\).  The block dimension is the number of compositions
of \(N\) into four parts, independent of \(n\) and \(K\).
\end{proof}

For \(N=1\) the ordered basis \((r_0,r_z,r_+,r_-)\) reproduces the
Briegel--Englert--Ginzel one-atom recurrence; for any \(N\) the same
formulas produce the \(D_N\times D_N\) blocks with no case-specific
algebra.  The conceptual point is that permutation symmetry keeps the
many-atom sector finite without destroying the nearest-neighbor radial
structure of the one-atom problem.

\section{Continued-fraction solution, observables, spectrum}
\label{sec:cf}

\subsection{Stationary recurrence and linear closure}
\label{sec:closure}

The global \(U(1)\) covariance of Eq.~(\ref{eq:master}) and the uniqueness
of the stationary state (Sec.~\ref{sec:model}) place it in the sector \(K=0\).  Suppressing the
sector label,
\begin{equation}
M_nX_n+G_nX_{n+1}+F_nX_{n-1}=0 .
\label{eq:stationary}
\end{equation}
With transfer matrices \(X_{n+1}=R_nX_n\), the stationary equation at
\(n+1\) gives
\begin{equation}
\big(M_{n+1}+G_{n+1}R_{n+1}\big)R_n=-F_{n+1},
\label{eq:riccati}
\end{equation}
so that downward iteration from a terminal \(R_{n_{\max}}\) determines
\(R_0\), \((M_0+G_0R_0)X_0=0\) determines \(X_0\) up to a constant
(numerically, \(X_0\) is the right singular vector of \(M_0+G_0R_0\)
belonging to its smallest singular value), and
\(X_{n+1}=R_nX_n\) reconstructs the state.  Normalization is elementary:
since \(\operatorname{tr}\rho_n^{(k)}=\delta_{n0}\delta_{k0}\) and
\(\operatorname{tr}R_{\mv}=\delta_{\mv,(N,0,0,0)}\), the trace of
\(\rho^{(0)}\) is the single component of \(X_0\) along
\(\Phi_{0,(N,0,0,0)}\), which is set to one.

For all \(n\ge0\) the blocks are affine in \(n\),
\(M_n=nM^{(1)}+M^{(0)}\), \(G_n=nG^{(1)}+G^{(0)}\), \(F_n=nF^{(1)}+F^{(0)}\)
(with \(F_0=0\) consistent with \(F^{(0)}=0\), Proposition~\ref{prop:grading}),
so a finite limit \(R_\infty=\lim_nR_n\), if it exists, satisfies the
algebraic Riccati equation
\begin{equation}
\big(M^{(1)}+G^{(1)}R_\infty\big)R_\infty+F^{(1)}=0 .
\label{eq:Rinf}
\end{equation}
Selecting a root of such a quadratic matrix equation is the usual
difficulty in terminating a matrix continued fraction
\cite{Risken1989}.  Here the equation is linear, and solvable throughout the dissipative
regime \(A>0\).

\begin{proposition}[Linear closure, solvable for \(A>0\)]
\label{prop:linear-closure}
(i) \(G_n=G^{(0)}\) for all \(n\ge0\), so \(G^{(1)}=0\) and
Eq.~(\ref{eq:Rinf}) reduces to \(M^{(1)}R_\infty=-F^{(1)}\).
(ii) In every sector \(K\),
\begin{align}
M^{(1)}&=-A\,\mathbb{I}+\frac{ig}{2}\,\mathcal N_K,\notag\\
\mathcal N_K&=[S_+,\,\cdot\,]\,\mathcal P_{q<K}+[S_-,\,\cdot\,]\,\mathcal P_{q>K},
\label{eq:M1-structure}
\end{align}
where \(\mathcal P_{q\lessgtr K}\) project onto the atomic basis states with charge
below or above \(K\) and \([S_\pm,\cdot]\) denotes the commutator action
\(A_\pm^{\rm L}-A_\pm^{\rm R}\), while \(\mathbb I\) is the identity on the
sector.  The operator \(\mathcal N_K\) is
nilpotent, \(\mathcal N_K^{\,N+|K|+1}=0\), so for every \(A>0\) and all
values of \(g,s,\Delta,B,C,\nu\), \(M^{(1)}\) is invertible and
\begin{equation}
R_\infty=-\big(M^{(1)}\big)^{-1}F^{(1)}
=\frac1A\sum_{j=0}^{N+|K|}\Big(\frac{ig}{2A}\Big)^{j}\mathcal N_K^{\,j}\,F^{(1)} .
\label{eq:Rinf-linear}
\end{equation}
\end{proposition}
\begin{proof}
(i) \(G_n=\mathcal V_{n+1}^{(-1)}\) depends on \(n\) only through the
\(\delta=-1\) field coefficients, and every \(\delta=-1\) coefficient in
Eqs.~(\ref{eq:fa1})--(\ref{eq:fc4}) is the constant \(1+\nu\).
(ii) The \(n\)-linear part of \(M_n\) comes from \(-An\) in
Eq.~(\ref{eq:uncoupled-diagonal}) and from the four \(\delta=0\)
coefficients that grow with \(n\): \(f^{\rm L}_{-,0}=f^{\rm R}_{-,0}=n+k\)
for \(k>0\) and \(f^{\rm L}_{+,0}=f^{\rm R}_{+,0}=n+|k|\) for \(k<0\), each
with unit coefficient of \(n\) and independent of \(\nu\); the \(k=0\)
coefficients are \(n\)-independent.  Inserting these into
Eq.~(\ref{eq:Vdelta}) and using \(k_\beta=K-q_\beta\) gives
Eq.~(\ref{eq:M1-structure}): for \(q_\beta<K\) the source couples through
\(a\), hence through \(S_+\) acting from the left minus from the right,
and for \(q_\beta>K\) through \(a^\dagger\), hence \([S_-,\cdot]\).
\([S_+,\cdot]\) raises the charge of every basis state by one and
\([S_-,\cdot]\) lowers it, so \(\mathcal N_K\) maps charge \(q<K\) to
\(q+1\), charge \(q>K\) to \(q-1\), and annihilates charge \(q=K\): each
application moves every component strictly closer to \(q=K\), where it is
killed, and since \(|q-K|\le N+|K|\) the operator vanishes after at most
\(N+|K|+1\) applications.  A nilpotent perturbation of \(-A\mathbb I\) is
invertible with the finite Neumann series of Eq.~(\ref{eq:Rinf-linear}).
\end{proof}

We have checked Eq.~(\ref{eq:M1-structure}) and the nilpotency
numerically for \(N\le6\) and \(K\in\{0,\pm1,2\}\) (agreement to
\(10^{-14}\), all eigenvalues of \(\mathcal N_K\) zero, index exactly
\(N+|K|+1\) at \(K=0\)).  The inverse is
bounded by a polynomial in \(g/2A\) of degree \(N+|K|\); accordingly the
smallest singular value of \(M^{(1)}\) decreases with coupling and atom
number (Table~\ref{tab:closure}: from 0.55 at \(g/A=0.3\) to 0.011 at
\(g/A=2.4\) for \(N=3\), and to \(1.4\times10^{-4}\) at \(g/A=4\), \(N=4\))
but never vanishes.  In practice we evaluate \(M^{(1)},F^{(1)}\) by finite
differences at a bulk index, set \(R_{n_{\max}}=R_\infty\), and recurse
down.

\subsection{Convergence and branch selection}
\label{sec:branch}

Proposition~\ref{prop:linear-closure} identifies the unique bounded
asymptotic transfer matrix.  Two further questions must be answered
before the output of the recursion can be called the stationary state:
whether the downward recursion started from \(R_\infty\) converges as
\(n_{\max}\to\infty\) and returns the normalizable (minimal) solution of
the three-term recurrence rather than a dominant one with factorially
growing \(X_n\), and in what sense the result is the stationary state
of the infinite-dimensional master equation rather than of a sequence of
finite truncations.  Both are settled here for a zero-temperature cavity,
under hypothesis (U) and a pivot condition (P) stated below.  The ingredients are a grading of the blocks, a
contraction property of the backward recursion, and a uniform
exponential moment bound for the stationary state.

Call \(|k_\beta|=|K-q_\beta|\) the \emph{grade} of the basis state
\(\beta\); it ranges over \(0,\ldots,N+|K|\).

\begin{proposition}[Grading and nilpotency]
\label{prop:grading}
In every sector \(K\) and for every \(\nu\ge0\): (i) \(F_n=nF^{(1)}\)
exactly, i.e.\ \(F^{(0)}=0\); (ii) \(F^{(1)}\) and \(\mathcal N_K\) lower
the grade of every basis state by exactly one and \(G\) raises it by
exactly one; (iii) consequently \(F^{(1)}\), \(G\), \(\mathcal N_K\), and
\(R_\infty=-(M^{(1)})^{-1}F^{(1)}\) are nilpotent, of index at most
\(N+|K|+1\).
\end{proposition}
\begin{proof}
By Eqs.~(\ref{eq:fa1})--(\ref{eq:fc4}) a \(\delta=+1\) coefficient occurs
only in \(a\rho_n^{(k)}\) and \(\rho_n^{(k)}a\) for \(k>0\) and in
\(a^\dagger\rho_n^{(k)}\) and \(\rho_n^{(k)}a^\dagger\) for \(k<0\); each
equals \(n+1\) times an \(n\)-independent constant and each lowers
\(|k|\) by one.  Hence \(F_n=\mathcal V_{n-1}^{(+1)}=nF^{(1)}\) with
\(F^{(1)}\) grade lowering.  A \(\delta=-1\) coefficient occurs in the
complementary four products for \(k\neq0\) and in all four products at
\(k=0\); each equals \(1+\nu\) and each raises \(|k|\) by one, so \(G\) is
grade raising.  The \(n\)-linear \(\delta=0\) coefficients are the
grade-lowering ones identified in the proof of
Proposition~\ref{prop:linear-closure}.  A product of \(N+|K|+1\)
grade-lowering (or grade-raising) maps vanishes, and
\(R_\infty=A^{-1}\sum_j(ig/2A)^j\mathcal N_K^{\,j}F^{(1)}\) lowers the
grade by at least one.
\end{proof}

Nilpotency of \(R_\infty\) means that along the bounded branch the
coefficients do not merely stay bounded: once \(R_n\to R_\infty\), the
product of any \(N+|K|+1\) consecutive transfer matrices tends to zero,
and with the rate \(R_n-R_\infty=O(1/n)\) established in
Proposition~\ref{prop:terminal} below that product is \(O(1/n)\), so
\(X_n\) decays faster than any geometric sequence.  This is
the behavior expected of the physical solution at \(\nu=0\), whose
coefficients are normally ordered moments (Lemma~\ref{lem:expmoments}
below).

Write \(\Phi_n(R)=-(M_{n+1}+GR)^{-1}F_{n+1}\) for the backward Riccati map
of Eq.~(\ref{eq:riccati}), so that \(R_n=\Phi_n(R_{n+1})\), and
\(R_n^{(n_{\max})}\) for the transfer matrices obtained by recursion from
level \(n_{\max}\).

\begin{proposition}[Terminal-condition independence]
\label{prop:terminal}
Fix \(\rho_0>\|R_\infty\|\).  There are \(n_0\) and \(C\), depending on
the parameters, \(N\), \(K\), and \(\rho_0\), such that for \(n\ge n_0\)
the map \(\Phi_n\) sends the ball \(\|R\|\le\rho_0\) into itself,
\(\|\Phi_n(R)-R_\infty\|\le C/(n+1)\), and
\(\|\Phi_n(R)-\Phi_n(R')\|\le C\|R-R'\|/(n+1)\).  Consequently two
recursions from terminal matrices in the ball differ at level
\(n_0\le n<n_{\max}\) by at most \(2\rho_0C^{\,n_{\max}-n}\,n!/n_{\max}!\),
the limits \(R_n^\infty=\lim_{n_{\max}\to\infty}R_n^{(n_{\max})}\) exist
for \(n\ge n_0\) and are independent of the terminal matrix, and
\(R_n^\infty=R_\infty+O(1/n)\).  If, in addition,
\begin{itemize}
\item[(P)] the pivots \(M_{j}+GR_{j}^\infty\) are invertible for all
\(j\ge1\) (for \(j>n_0\) this is automatic),
\end{itemize}
the same holds for all \(n\ge0\).
\end{proposition}
\begin{proof}
Write \(M_{n+1}+GR=(n+1)[M^{(1)}+\epsilon_n(R)]\) with
\(\epsilon_n(R)=(M^{(0)}+GR)/(n+1)\).  For \(\|R\|\le\rho_0\) and
\(n+1\ge2\|(M^{(1)})^{-1}\|(\|M^{(0)}\|+\|G\|\rho_0)\) the Neumann series
gives \(\|(M^{(1)}+\epsilon_n)^{-1}\|\le2\|(M^{(1)})^{-1}\|\).  Since
\(F_{n+1}=(n+1)F^{(1)}\),
\(\Phi_n(R)-R_\infty=-(M^{(1)}+\epsilon_n)^{-1}\epsilon_nR_\infty\), of
norm at most \(C_1/(n+1)\) with
\(C_1=2\|(M^{(1)})^{-1}\|(\|M^{(0)}\|+\|G\|\rho_0)\|R_\infty\|\); for
\(n+1\ge C_1/(\rho_0-\|R_\infty\|)\) the ball is mapped into itself.
Next, \(\Phi_n(R)-\Phi_n(R')=(M_{n+1}+GR)^{-1}G(R'-R)\Phi_n(R')\), of norm
at most \(C\|R-R'\|/(n+1)\) with
\(C_2=2\|(M^{(1)})^{-1}\|\,\|G\|\rho_0\).  Taking \(n_0\) above both
thresholds and \(C=\max(C_1,C_2)\), two recursions from the ball differ at level
\(n\) by at most \(\prod_{j=n+1}^{n_{\max}}(C/j)\) times their difference
at level \(n_{\max}\), which is at most \(2\rho_0\); comparing the
recursions from \(n_{\max}\) and \(n_{\max}'>n_{\max}\) (the latter is in
the ball at level \(n_{\max}\)) shows that \(R_n^{(n_{\max})}\) is Cauchy
in \(n_{\max}\), with limit independent of the terminal matrix, and the
first estimate gives \(R_n^\infty=R_\infty+O(1/n)\).  Below \(n_0\) the
limits are obtained by the finitely many maps \(\Phi_n\), which are
continuous at the limit points when (P) holds.
\end{proof}

The constants are not sharp: the bound uses operator norms of nilpotent
matrices whose Neumann series are large, and the actual recursion forgets
its terminal matrix much faster (a random terminal matrix of unit norm at
\(n_{\max}=20\) changes \(R_0\) by less than \(10^{-15}\) at \(N=3\);
Appendix~\ref{app:validation}).  The proposition states that the
\emph{limit} of the recursion is well defined and that the closure
\(R_\infty\) is a choice of terminal condition that affects only the rate
of convergence, not the result; Miller's termination
\(X_{n_{\max}+1}=0\) is the choice \(R_{n_{\max}}=0\).

The physical solution is characterized through the dual basis.  At
\(\nu=0\) the left eigenoperators of \(\mathcal L_{\rm f}\) dual to the
\(\rho_n^{(k)}\) are the normally ordered monomials
\(\check\rho_n^{(k)}=a^{\dagger n}a^{n+k}/(n+k)!\) for \(k\ge0\) and
\(a^{\dagger\,n+|k|}a^{n}/(n+|k|)!\) for \(k<0\)
[\(\Tr(\check\rho_m^{(l)}\rho_n^{(k)})=\delta_{mn}\delta_{lk}\), checked to
\(10^{-12}\) in a 40-level Fock space], and with \(\check R_\alpha\) the
dual atomic basis the coefficients of any trace-class \(\rho\) with finite
moments are \(X_{n\alpha}=\Tr(\rho\,\check\Phi_{n\alpha})\),
\(\check\Phi_{n\alpha}=\check\rho_n^{(k_\alpha)}\otimes\check R_\alpha\).
Because \(\mathcal L^\dagger\check\Phi_{m\beta}\) is a finite combination
of dual basis elements with the coefficients of
Eq.~(\ref{eq:blockjacobi}), the stationary recurrence
(\ref{eq:stationary}) for the moment sequence of \(\rho\) is literally
the set of stationary moment equations
\(\Tr(\rho\,\mathcal L^\dagger\check\Phi_{m\beta})=0\): no interchange of
infinite sums is involved.

\begin{lemma}[Uniform exponential moment bound]
\label{lem:expmoments}
Let \(\theta>0\) with \(e^{\theta}<(1+\nu)/\nu\) (any \(\theta>0\) if
\(\nu=0\)), and let \(\rho\) be the stationary state of
Eq.~(\ref{eq:master}) truncated at any cavity cutoff \(n_{\rm c}\ge1\).
Then
\begin{equation}
\langle e^{\theta\Exc}\rangle_\rho\le2e^{\theta(N+n_*)},\qquad
n_*=\Big\lceil\frac{2(BsN+A\nu)}{A[(1+\nu)e^{-\theta}-\nu]}\Big\rceil ,
\label{eq:expmoment}
\end{equation}
independently of \(n_{\rm c}\).
\end{lemma}
\begin{proof}
\(\Exc\) commutes with \(H\), with \(\nop\), and with the Fock projector
defining the truncation, so the Hamiltonian does not contribute to
\(\mathcal L^\dagger f(\Exc)\), and each jump operator shifts \(\Exc\) by
\(0\) or \(\pm1\): \(\mathcal L^\dagger e^{\theta\Exc}=e^{\theta\Exc}W\)
with
\(W=(e^{-\theta}-1)[A(1+\nu)\nop+B(1-s)N_e]+(e^{\theta}-1)[A\nu(\nop+1)+Bs(N-N_e)]\)
(in the truncated space the up-transition from the top level is absent,
which only lowers \(W\)).  With \(u=e^{\theta}-1\),
\(W\le-uA[(1+\nu)e^{-\theta}-\nu]\nop+u(BsN+A\nu)\le-\tfrac12uA[(1+\nu)e^{-\theta}-\nu]\nop\)
whenever \(\nop\ge n_*\).  Stationarity gives
\(\langle e^{\theta\Exc}W\rangle=0\), hence
\(\tfrac12uA[(1+\nu)e^{-\theta}-\nu]\,n_*\langle e^{\theta\Exc}\mathbf 1_{\nop\ge n_*}\rangle
\le u(BsN+A\nu)\langle e^{\theta\Exc}\mathbf 1_{\nop<n_*}\rangle\le
u(BsN+A\nu)e^{\theta(n_*-1+N)}\), so that
\(\langle e^{\theta\Exc}\mathbf 1_{\nop\ge n_*}\rangle\le e^{\theta(n_*+N)}\),
and \(\langle e^{\theta\Exc}\mathbf 1_{\nop<n_*}\rangle\le e^{\theta(n_*+N)}\)
trivially.
\end{proof}

At \(\nu=0\) the bound holds for every \(\theta\) (with \(n_*\) growing
like \(e^{\theta}\), so the bound is doubly exponential in \(\theta\) but
finite), so the photon
distribution has tails \(P(\nop\ge m)\le2e^{\theta(N+n_*)}e^{-\theta m}\)
for all \(\theta\), the factorial-moment generating function
\(\sum_nc_nz^n=\langle(1+z)^{\nop}\rangle\) is entire, and \(c_nr^n\to0\)
for every \(r\); by the Cauchy--Schwarz inequality
\(|X_{n\alpha}|\le\|\check R_\alpha\|\,(c_nc_{n+|k_\alpha|})^{1/2}\),
every component of the physical \(X_n\) decays super-exponentially.  For
\(\nu>0\) the admissible \(\theta\) are bounded by the thermal tail of the
bath, as they must be.

\begin{proposition}[Convergence of Fock truncations]
\label{prop:fock}
Let \(\nu=0\) and assume (U).  Then \(\rho^{(n_{\rm c})}\to\rho_{\rm ss}\)
in trace norm as \(n_{\rm c}\to\infty\), \(\rho_{\rm ss}\in\mathcal
T_{\exp}\) obeys Eq.~(\ref{eq:expmoment}), and
\(\langle O\rangle_{n_{\rm c}}\to\langle O\rangle_{\rm ss}\) for every
observable with \(\|Oe^{-\theta'\Exc}\|<\infty\) for some \(\theta'>0\),
in particular for every polynomial in \(a\), \(a^\dagger\), and the
atomic operators, hence for every stationary observable reported in this
paper.
\end{proposition}
\begin{proof}
By Lemma~\ref{lem:expmoments} the family \(\{\rho^{(n_{\rm c})}\}\) has
uniformly bounded \(\langle\Exc\rangle\), and \(\Exc\) has
finite-dimensional eigenspaces, so the family is tight; every subsequence
has a further subsequence converging in the weak operator topology to a
positive trace-class \(\rho_*\) with \(\Tr\rho_*=1\) (no mass escapes),
and for density operators weak convergence to a density operator implies
trace-norm convergence \cite{Dellantonio1967}.  Fatou's lemma transfers
Eq.~(\ref{eq:expmoment}) to \(\rho_*\), so \(\rho_*\in\mathcal T_{\exp}\).
For an observable \(Y\) supported on finitely many Fock levels,
\(\mathcal L^\dagger_{n_{\rm c}}Y=\mathcal L^\dagger Y\) once \(n_{\rm c}\)
exceeds the support by one, so
\(\Tr(\rho_*\mathcal L^\dagger Y)=\lim\Tr(\rho^{(n_{\rm c})}\mathcal L^\dagger_{n_{\rm c}}Y)=0\);
for polynomial \(Y\) the same follows by dominated convergence from its
Fock-truncated approximants, using the exponential moments.  By (U),
\(\rho_*=\rho_{\rm ss}\), and since every subsequence has a sub-subsequence
with this limit, the whole sequence converges.  For the observables,
split \(O=O\mathbf 1_{\Exc\le M}+O\mathbf 1_{\Exc>M}\): the first part is
bounded and converges by trace-norm convergence, the second is bounded
uniformly in \(n_{\rm c}\) by
\(\|Oe^{-\theta'\Exc}\|\,2e^{\theta(N+n_*)}e^{-(\theta-\theta')M}\) with
\(\theta>\theta'\), which is arbitrarily small for large \(M\).
\end{proof}

\begin{theorem}[The continued fraction converges to the stationary state]
\label{thm:branch}
Let \(\nu=0\), \(K=0\), and assume (U) and (P).  Then
\(\ker(M_0+GR_0^\infty)\) is one-dimensional; the normalized solution
\(X^{(n_{\max})}\) of the continued fraction converges as
\(n_{\max}\to\infty\), component by component, to the damping-basis
coefficients of the unique stationary state,
\(\lim_{n_{\max}\to\infty}X^{(n_{\max})}_{n\alpha}=\Tr(\rho_{\rm ss}\check\Phi_{n\alpha})\);
and every observable functional of Sec.~\ref{sec:observables} evaluated
on \(X^{(n_{\max})}\) converges to its stationary value, ratios such as
\(g^{(2)}(0)\) and the sector-conditional quantities of
Sec.~\ref{sec:sectors} whenever the limiting denominator is nonzero.
The statement is one of componentwise coefficient convergence and
convergence of these functionals; no operator-norm convergence of a
reconstructed density operator is claimed.
\end{theorem}

The proof is given in Appendix~\ref{app:proof}.  Its logic is: the
physical moment sequence solves the recurrence and decays
super-exponentially (Lemma~\ref{lem:expmoments},
Proposition~\ref{prop:fock}); any super-exponentially decaying solution
of the recurrence reconstructs, through the finite-rank operators
\(\rho_n^{(k)}\), to an element of \(\mathcal T_{\exp}\) that satisfies
the moment equations and is therefore a multiple of \(\rho_{\rm ss}\) by
(U), so the space of such solutions is one-dimensional; and the recursion
of Proposition~\ref{prop:terminal} produces a super-exponentially
decaying solution, because its transfer matrices converge to the
nilpotent \(R_\infty\), while the physical solution is shown to satisfy
\(X_{n+1}=R_n^\infty X_n\) by the same contraction estimate.  The theorem
is the analogue, for this matrix recurrence, of the scalar Pincherle
theorem reviewed in Ref.~\cite{Gautschi1967}, established here directly
rather than imported: the bounded branch of the matrix continued
fraction is the minimal solution, and the minimal solution is the
state.

\begin{remark}[What remains numerical]
\label{rem:status}
Theorem~\ref{thm:branch} covers the stationary sector \(K=0\) at
\(\nu=0\) under (U) and (P).  It does not cover: (i) the finite radial
truncation and floating-point arithmetic of an actual computation, which
are the approximations of the method and are quantified in
Sec.~\ref{sec:precision} and Appendix~\ref{app:validation}; (ii) a thermal
cavity, \(\nu>0\), for which Propositions~\ref{prop:grading} and
\ref{prop:terminal} still hold but the identification step uses the
super-exponential decay specific to \(\nu=0\) --- the thermal row of
Table~\ref{tab:closure} is validated against brute force only; (iii) the \(K=-1\) regression computation of
Sec.~\ref{sec:spectrum}, a finite-truncation eigenvalue problem whose
convergence is established numerically
(Table~\ref{tab:spectral-checks}); (iv) hypothesis (P), which concerns
the \emph{limiting} pivots \(M_j+GR_j^\infty\) and remains an
assumption: a computation inverts the finite-cutoff pivots
\(P_j^{(n_{\max})}=M_j+GR_j^{(n_{\max})}\), whose invertibility and
conditioning we monitor (Sec.~\ref{sec:precision},
Appendix~\ref{app:branch-checks}).  A sufficient perturbative criterion
for the invertibility of the limiting pivot would be
\(\|(P_j^{(n_{\max})})^{-1}G\|\,\|R_j^{(n_{\max})}-R_j^\infty\|<1\); the
distance to the limiting transfer matrix is not available, and the
corresponding diagnostic, evaluated with the difference between two
finite-cutoff transfer matrices as an estimate of it, is small for the
tested low-index pivots at \(N\le8\).  This supports, but does not
establish, hypothesis (P): two finite approximations can agree while
sharing a common error.  At \(N=16\) the diagnostic exceeds the
threshold in Euclidean norm, which does not imply that the limiting
pivot is singular but marks the regime in which the conditioning of the
pivots is the practical limitation discussed in
Sec.~\ref{sec:precision}.  The three
limits that enter the method --- \(n_{\rm c}\to\infty\) in the brute-force
comparisons (Proposition~\ref{prop:fock}), \(n_{\max}\to\infty\) in the
recursion (Proposition~\ref{prop:terminal} and
Theorem~\ref{thm:branch}), and the choice of asymptotic closure
(Proposition~\ref{prop:terminal}) --- are thus distinct, and the first two
converge to the same object.
\end{remark}

\subsection{Cost and conditioning}
\label{sec:precision}

The cost of the stationary solve is \(O(n_{\max}D_N^3)\), one dense
\(D_N\times D_N\) solve per radial index, i.e.\ \(O(N^9)\) for the dense
blocks used here; exploiting the sparsity of the blocks would reduce this.
In the baseline regime of Sec.~\ref{sec:collective} the solve takes
about 0.02~s at \(N=3\), 0.5~s at \(N=8\), and 25--40~s at \(N=16\)
(\(D_{16}=969\), \(n_{\max}=50\)) on a single laptop core, against a
full Liouvillian of dimension \(4^Nn_{\rm cav}^2\) that is out of reach of
exact diagonalization beyond \(N\approx5\).  The radial truncation must exceed
the support of the state, \(n_{\max}\gtrsim\langle\nop\rangle+\text{few}\sqrt{\langle\nop\rangle}\),
since at \(\nu=0\) the coefficients are the normally ordered factorial
moments, \(c_n=\langle a^{\dagger n}a^n\rangle/n!\), which for a
near-Poissonian field of mean \(\mu\) are \(\mu^n/n!\), peaking near
\(n=\mu\) at \(\approx e^{\mu}/\sqrt{2\pi\mu}\).

That growth is also an empirical double-precision conditioning limit of
the present representation and implementation: a state with \(\mu\)
photons is represented by coefficients spanning a factor \(\sim e^{\mu}\)
(about \(\mu/\ln10\) decimal orders) relative to the normalization
component, which suggests a floating-point error scale of order
\(\epsilon\,e^{\mu}\), \(\epsilon\approx10^{-16}\).  Table~\ref{tab:precision}
is consistent with this estimate in the good-cavity regime of
Sec.~\ref{sec:lasing}.  Its ``spread'' column is a stability diagnostic
(the variation of the result over four radial truncations), not a
certified error; the actual double-precision errors, obtained by
repeating the recursion in 40-digit arithmetic
(Appendix~\ref{app:branch-checks}), are of the same order:
the absolute errors of \(\langle\nop\rangle\) over \(n_{\max}\in\{60,90,120,160\}\) are \(6\times10^{-15}\)--\(7\times10^{-14}\) at \(N=2\), \(1\times10^{-11}\)--\(3\times10^{-11}\) at \(N=3\), \(1\times10^{-10}\)--\(7\times10^{-9}\) at \(N=4\), \(2\times10^{-6}\)--\(6\times10^{-6}\) at \(N=5\), and \(1.1\)--\(1.5\times10^{-3}\) at \(N=6\), against the spreads \(10^{-13}\), \(3\times10^{-11}\), \(10^{-8}\), \(10^{-5}\), \(10^{-2}\) of the table.  With these reference values, 13 digits of
\(\langle\nop\rangle\) are correct at \(\mu=6\), 10 at \(\mu=10\),
9 at \(\mu=14\), 6 at \(\mu=18\), and 4 at \(\mu=22\), and every \(g^{(2)}(0)\) digit quoted in
Sec.~\ref{sec:lasing} for \(N\le5\) is correct.  A direct sparse solve
of the truncated block-Jacobi system in place of the downward recursion,
and a diagonal rescaling of the radial coefficients by Poissonian
weights \(\mu^n/n!\) (which changes the double-precision results by less
than the spread), do not improve this; we have not explored other
preconditionings, and the loss is therefore attributed to the
representation as implemented, not proved to be intrinsic.  In double
precision the practical high-accuracy regime is roughly
\(\langle\nop\rangle\lesssim20\), the precise limit depending on the
observable and tolerance --- the few-quanta regime in which microlasers
operate --- and Fock-truncated permutation-invariant solvers remain the
tool of choice at larger photon numbers.

A second, independent loss of precision grows with the atom number
rather than the photon number.  The Neumann series of
Eq.~(\ref{eq:Rinf-linear}) has large norm --- \(\|(M^{(1)})^{-1}\|=1.5\),
\(3.2\), \(11\), \(50\), \(2.9\times10^{3}\) for \(N=1,2,3,4,6\) at the
baseline parameters --- because the nilpotent \(\mathcal N_K\) is far
from normal, and the Euclidean condition numbers of the pivots
\(M_{n+1}+GR_{n+1}\) deteriorate accordingly (reciprocal condition
numbers of order \(10^{-2}\) at \(N=3\), \(10^{-10}\) at \(N=8\), and
\(10^{-15}\)--\(10^{-29}\) at \(N=16\); Appendix~\ref{app:branch-checks}),
as does the gap between the smallest and the second-smallest singular
value of \(M_0+G_0R_0\) that separates the stationary kernel vector from
the rest.  The consequence is visible in the reproducibility of the
observables under a change of \(n_{\max}\) or of terminal condition
(Table~\ref{tab:Nconditioning}): relative differences of
\(10^{-15}\)--\(10^{-12}\) at \(N=8\), up to \(6\times10^{-8}\) at
\(N=12\), \(5\times10^{-8}\) at \(N=16\), \(s=0.6\), and
\(6\times10^{-4}\) for \(\langle\nop\rangle\) and \(g^{(2)}(0)\) and
\(4\times10^{-3}\) for \(C_2^{(N)}\) at \(N=16\),
\(s=0.9\), although \(\langle\nop\rangle\approx4\) there (the \(s=0.6\)
and the \(\lambda=0.1\) runs quoted here are stored with the output of
Ref.~\cite{code} alongside Table~\ref{tab:Nconditioning}).
Reproducibility is moreover a necessary, not a sufficient, indicator:
at \(N=12\), \(s=0.3\), the plain recursion is reproducible to
\(10^{-13}\) but differs from the better-conditioned rescaled
computation described next by \(10^{-10}\).  At \(N=16\)
the transfer matrices themselves are not reproducible between
\(n_{\max}=40\) and \(50\) --- only their action on the stationary
solution is --- so for the values plotted for \(N\ge12\) in
Fig.~\ref{fig:Nscaling} the plain recursion exhibits a cutoff and
terminal-condition sensitivity corresponding to roughly three
significant digits in \(\langle\nop\rangle\) and \(g^{(2)}(0)\), and two
to three in \(C_2^{(N)}\), at the strongest-pump points.  This loss, unlike the
\(e^{\mu}\) one, responds to preconditioning: a grade-dependent diagonal
rescaling \(\Phi_{n\alpha}\to\lambda^{|k_\alpha|}\Phi_{n\alpha}\), which
is exact in exact arithmetic and leaves the \(|k|=0\) components (hence
\(\langle\nop\rangle\), \(g^{(2)}(0)\)) unchanged, balances the
grade-raising and grade-lowering blocks.  With \(\lambda=0.3\) the
cutoff sensitivity at \(N=16\), \(s=0.9\) is reduced to
\(6\times10^{-9}\), \(5\times10^{-9}\), and \(4\times10^{-8}\) for
\(\langle\nop\rangle\), \(g^{(2)}(0)\), and \(C_2^{(N)}\)
(Table~\ref{tab:Nconditioning}), the rescaled values agree with the
plain ones within the latter's sensitivity, and the results for \(\lambda=0.2\),
\(0.3\), and \(0.4\) at two cutoffs agree among themselves to \(5\times10^{-7}\), \(6\times10^{-7}\), and \(3\times10^{-6}\) for
\(\langle\nop\rangle\), \(g^{(2)}(0)\), and \(C_2^{(N)}\) at \(s=0.9\)
(\(10^{-15}\) at \(s=0.3\)), with kernel gaps \(\sigma_2\) between
\(10^{-7}\) and \(4\times10^{-6}\).
This is a stability result: the rescaled values are reproducible, but
their absolute accuracy at \(N=16\) has not been established by an
independent (higher-precision or differently formulated) reference,
which is beyond the reach of the present implementation at
\(D_{16}=969\).  The gap of \(M_0+G_0R_0\) must be
monitored: too strong a rescaling (\(\lambda=0.1\)) collapses it below
the floating-point resolution, and the kernel vector returned by the
singular-value decomposition is then wrong by 1\% while still
reproducible under \(n_{\max}\), so reproducibility alone is not a
sufficient diagnostic.  All numbers quoted in this paper for
\(N\le8\) are unaffected by this issue at the digits shown; the \(N=16\)
timing and feasibility statements stand, with the precision caveat
just given.

\begin{table*}[tbp]
\caption{Conditioning of the \(\nu=0\) damping-basis continued fraction at
large photon number.  Good-cavity parameters \(A=0.1\), \(B=1\), \(C=B/2\),
\(g=0.5\), \(s=0.9\).  \(\max_n|c_n|\) is the largest radial coefficient of
the trace-normalized state; the spread is the variation of
\(\langle\nop\rangle\) over \(n_{\max}\in\{60,90,120,160\}\).  The \(N=2,3\)
values agree with brute-force diagonalization (cavity cutoff 40 and 55) to
all digits shown, and all values shown agree with the 40-digit reference
computation of Appendix~\ref{app:branch-checks} (for \(N=6\) the
double-precision recursion itself gives 22.00 and 1.0087, in error by
\(1.5\times10^{-3}\) and \(7\times10^{-5}\)); the last two columns are the
variations of \(\langle\nop\rangle\) and of \(g^{(2)}(0)\) over the same
set of truncations, a stability diagnostic.}
\label{tab:precision}
\begin{ruledtabular}
\begin{tabular}{ccccccc}
\(N\) & \(\langle\nop\rangle\) & \(g^{(2)}(0)\) & \(\max_n|c_n|\) & \(\epsilon\max|c_n|\) & spread of \(\langle\nop\rangle\) & spread of \(g^{(2)}(0)\)\\
\hline
2 & 6.0907507 & 1.0375577 & \(9.5\times10^{1}\) & \(10^{-14}\) & \(10^{-13}\) & \(2\times10^{-14}\)\\
3 & 10.031203 & 1.0211641 & \(4.5\times10^{3}\) & \(10^{-12}\) & \(3\times10^{-11}\) & \(6\times10^{-12}\)\\
4 & 14.011935 & 1.0143123 & \(2.4\times10^{5}\) & \(10^{-11}\) & \(10^{-8}\) & \(10^{-9}\)\\
5 & 18.00157 & 1.010826 & \(1.4\times10^{7}\) & \(10^{-9}\) & \(10^{-5}\) & \(6\times10^{-7}\)\\
6 & 21.99 & 1.0087 & \(8.2\times10^{8}\) & \(10^{-7}\) & \(10^{-2}\) & \(5\times10^{-4}\)\\
\end{tabular}
\end{ruledtabular}
\end{table*}

\begin{table*}[tbp]
\caption{Reproducibility of the double-precision results at larger
\(N\), baseline parameters.  ``Plain'': relative spread of
\(\langle\nop\rangle\), \(g^{(2)}(0)\), and \(C_2^{(N)}\) over the
recursion from \(R_\infty\) with \(n_{\max}=40\) and \(50\) and from
Miller's termination with \(n_{\max}=40\); \(\sigma_2\) is the
second-smallest singular value of \(M_0+G_0R_0\) (the smallest is
zero to round-off).  ``Rescaled'': the spreads of \(\langle\nop\rangle\), \(g^{(2)}(0)\), and
\(C_2^{(N)}\) with the grade rescaling \(\lambda=0.3\) (\(n_{\max}=40\)
and \(50\)).  All entries are reproducibility measures, not accuracy
bounds.  The last column is the
relative difference between the plain and the rescaled
\(\langle\nop\rangle\) at \(n_{\max}=50\); the \(\langle\nop\rangle\)
column is the rescaled value at \(n_{\max}=50\).}
\label{tab:Nconditioning}
\begin{ruledtabular}
\begin{tabular}{ccccccccccc}
\(N\) & \(s\) & \(\langle\nop\rangle\) & \multicolumn{3}{c}{plain: spread of \(\langle\nop\rangle\), \(g^{(2)}(0)\), \(C_2\)} & \(\sigma_2\) & \multicolumn{3}{c}{rescaled (\(\lambda=0.3\)): spread of \(\langle\nop\rangle\), \(g^{(2)}(0)\), \(C_2\)} & plain vs.\ rescaled\\
\hline
8 & 0.3 & 0.22276 & \(3\times10^{-15}\) & \(3\times10^{-15}\) & \(2\times10^{-15}\) & \(2\times10^{-2}\) & \(10^{-15}\) & \(2\times10^{-15}\) & \(10^{-15}\) & \(2\times10^{-15}\)\\
8 & 0.9 & 2.06187 & \(7\times10^{-13}\) & \(8\times10^{-13}\) & \(8\times10^{-12}\) & \(3\times10^{-2}\) & \(4\times10^{-16}\) & \(10^{-14}\) & \(2\times10^{-14}\) & \(3\times10^{-13}\)\\
12 & 0.3 & 0.25033 & \(10^{-13}\) & \(10^{-15}\) & \(9\times10^{-14}\) & \(5\times10^{-5}\) & \(10^{-15}\) & \(10^{-15}\) & \(2\times10^{-15}\) & \(10^{-10}\)\\
12 & 0.9 & 3.10306 & \(9\times10^{-9}\) & \(10^{-8}\) & \(6\times10^{-8}\) & \(2\times10^{-4}\) & \(9\times10^{-13}\) & \(10^{-12}\) & \(6\times10^{-12}\) & \(10^{-9}\)\\
16 & 0.3 & 0.26685 & \(4\times10^{-10}\) & \(3\times10^{-10}\) & \(3\times10^{-10}\) & \(5\times10^{-8}\) & \(4\times10^{-16}\) & \(4\times10^{-16}\) & \(8\times10^{-16}\) & \(7\times10^{-11}\)\\
16 & 0.9 & 4.15606 & \(6\times10^{-4}\) & \(7\times10^{-4}\) & \(4\times10^{-3}\) & \(10^{-7}\) & \(6\times10^{-9}\) & \(5\times10^{-9}\) & \(4\times10^{-8}\) & \(10^{-4}\)\\
\end{tabular}
\end{ruledtabular}
\end{table*}

\subsection{Observables from \(X_0\)}
\label{sec:observables}

Because \(\operatorname{tr}R_{\mv}\) vanishes unless \(\mv=(N,0,0,0)\)
and \(\operatorname{tr}\rho_n^{(k)}=\delta_{n0}\delta_{k0}\), tracing the
stationary state over the field keeps only the radial-index-zero,
charge-neutral components:
\begin{equation}
\rho_N=\operatorname{tr}_{\rm field}\rho
=\sum_{\alpha:\,q_\alpha=0}(X_0)_\alpha R_{\mv_\alpha}.
\label{eq:rhoN-reconstruction}
\end{equation}
Every atomic observable then follows as an ordinary trace.  Joint
observables \(\langle a^\dagger O\rangle\) with \(O\) atomic of charge
\(-1\) (such as \(a^\dagger S_-\) or \(a^\dagger S_-P_J\)) follow in the
same way, since \(\operatorname{tr}(\rho_n^{(-1)}a^\dagger)=\delta_{n0}\)
by Eq.~(\ref{eq:fb4}):
\begin{equation}
\langle a^\dagger O\rangle=\sum_{\alpha:\,q_\alpha=+1}(X_0)_\alpha
\operatorname{tr}\big(R_{\mv_\alpha}O\big).
\label{eq:joint-observable}
\end{equation}
Photon-number moments come from the \((N,0,0,0)\) components \(c_n\) of
\(X_n\), \(\langle\nop\rangle=\nu+(1+\nu)c_1\) and
\(\langle\nop(\nop-1)\rangle=2\nu^2+4\nu(1+\nu)c_1+2(1+\nu)^2c_2\)
(Appendix~\ref{app:observables}).  Equations~(\ref{eq:rhoN-reconstruction})
and (\ref{eq:joint-observable}) are verified against brute-force partial
traces to better than \(10^{-15}\) per matrix element.

\subsection{Regression and emission spectrum}
\label{sec:spectrum}

The first-order coherence
\(g^{(1)}(\tau)=\langle a^\dagger(\tau)a\rangle/\langle\nop\rangle\)
follows from the quantum regression theorem as
\(\operatorname{tr}[a^\dagger e^{\mathcal L\tau}(a\rho_{\rm ss})]/\langle\nop\rangle\).
The stationary state is permutation invariant, \(a\rho_{\rm ss}\) is
therefore permutation invariant in the emitter coordinates and lies in the
sector \(K=-1\), and the regression evolution cannot leave the invariant
operator subspace; so \(a\rho_{\rm ss}\) evolves under the \(K=-1\)
block-Jacobi operator of Proposition~\ref{prop:blockjacobi}, and any
non-symmetric Liouvillian modes --- for \(N=2\) the slowly decaying
singlet--triplet coherences --- have zero amplitude in this correlator.
The coefficients of \(a\rho_{\rm ss}\) follow from those of
\(\rho_{\rm ss}\) through the \(\mu=-1\) field identities, and the
functional \(\operatorname{tr}(a^\dagger\,\cdot\,)\) picks the
\((n{=}0,(N,0,0,0))\) component (Appendix~\ref{app:observables}).
Diagonalizing the \(K=-1\) block-Jacobi matrix at a converged radial
truncation gives the pole decomposition
\begin{equation}
g^{(1)}(\tau)=\sum_jw_je^{\lambda_j\tau},\qquad\sum_jw_j=1 ,
\label{eq:g1poles}
\end{equation}
in the frame rotating at \(\omega\).  The \(w_j\) are residues of a
non-normal operator: they are complex in general, and individual
residues can exceed unity in modulus while summing to one.  At
resonance, \(\Delta=0\), and up to the uniform frame shift \(-i\omega K\)
of Eq.~(\ref{eq:uncoupled-diagonal}), the \(K=-1\) block is similar,
through the diagonal phase transformation
\(\Phi_{n\alpha}\to i^{q_\alpha}\Phi_{n\alpha}\), to a real matrix, and the
initial vector \(a\rho_{\rm ss}\) and the functional \(\Tr(a^\dagger\,\cdot\,)\)
are real in the same transformed basis, so its complex poles and their
residues occur in conjugate pairs; away from resonance Hermiticity relates the \(K=-1\)
and \(K=+1\) sectors to each other and no pairing within one sector is
implied.  All spectra shown are at \(\Delta=0\).  The emission spectrum
\(S(\omega')\propto\operatorname{Re}\int_0^\infty d\tau\,e^{i\omega'\tau}g^{(1)}(\tau)
=\operatorname{Re}\sum_jw_j/(-\lambda_j-i\omega')\)
is a rational sum of pole contributions, Lorentzian plus dispersive, which
is real and, in all cases computed, positive; here and in
Fig.~\ref{fig:spectra} \(\omega'\) is the detuning from the cavity
frequency, since \(g^{(1)}(\tau)\) is taken in the frame rotating at
\(\omega\).  We characterize the line
by the pole \(\lambda_1\) with the largest real part among those whose
residue modulus (a conjugate pair counting jointly) exceeds a threshold
\(\theta\), and report \(-2\operatorname{Re}\lambda_1\) (the width of that
pole's Lorentzian) and \(|\operatorname{Im}\lambda_1|\) (its offset from
\(\omega\)).  Appendix~\ref{app:validation} shows that the selected pole
is the same for \(\theta\) between 1\% and 30\% in every case reported,
that \(\lambda_1\) is converged to \(2\times10^{-12}\) or better in \(n_{\max}\), that it
agrees with the corresponding pole of the brute-force Liouvillian to
\(10^{-13}\) (Table~\ref{tab:linewidth-validation}), and that in the
single-line regime of Sec.~\ref{sec:lasing} at \(s\ge0.6\) the full width
at half maximum of \(S(\omega')\) equals \(-2\operatorname{Re}\lambda_1\)
to 0.1--0.5\% (at \(s=0.3\), where secondary poles still carry residue,
the two differ by 3--9\%).  The FWHM is always the outermost half-height
span of the full spectral function, so for a resolved doublet it spans
both components and exceeds the width of either; and when a doublet
has just merged the two overlapping components with complex residues
can produce a single peak \emph{narrower} than either pole --- at the
baseline parameters, \(N=2\), \(s=0.9\), the FWHM is \(0.82A\) while
\(-2\operatorname{Re}\lambda_1=1.12A\).  The pole characterization and
the spectral function are therefore shown side by side:
Fig.~\ref{fig:spectra} gives the normalized \(S(\omega')\) in the three
regimes discussed in Sec.~\ref{sec:lasing}.

\begin{table*}[tbp]
\caption{Dominant pole \(\lambda_1\) of \(g^{(1)}(\tau)\) from the \(K=-1\)
continued fraction (\(n_{\max}=30\); identical at \(n_{\max}=40\) to all
digits shown) and from the brute-force Liouvillian (cavity cutoff 14),
with the normalized residue modulus \(|w_1|\) of the conjugate pair
(residues of a non-normal operator need not lie below one; Sec.~\ref{sec:spectrum}).  Baseline
parameters \(A=1\), \(B=0.7\), \(C=B/2\), \(g=1.1\), \(\Delta=0\).}
\label{tab:linewidth-validation}
\begin{ruledtabular}
\begin{tabular}{ccccc}
\(N\) & \(s\) & \(\lambda_1\) (continued fraction) & \(\lambda_1\) (brute force) & \(|w_1|\)\\
\hline
1 & 0.3 & \(-0.448708365\pm0.457879042i\) & \(-0.448708365\pm0.457879042i\) & 0.995\\
1 & 0.8 & \(-0.471261791\pm0.254376872i\) & \(-0.471261791\pm0.254376872i\) & 0.992\\
2 & 0.3 & \(-0.479710666\pm0.651980285i\) & \(-0.479710666\pm0.651980285i\) & 0.961\\
2 & 0.8 & \(-0.563533729\pm0.302524753i\) & \(-0.563533729\pm0.302524753i\) & 0.986\\
3 & 0.6 & \(-0.558125741\pm0.522869362i\) & --- & 1.018\\
\end{tabular}
\end{ruledtabular}
\end{table*}

\section{Collective correlations in two and three emitters}
\label{sec:collective}

Unless stated otherwise the baseline parameters are \(A=1\), \(B=0.7\),
\(C=B/2\), \(g=1.1\), \(\Delta=0\), \(\nu=0\): a strongly coupled,
bad-cavity (\(A>C\)) few-quanta regime with \(\langle\nop\rangle<1\).

\subsection{Exact two-emitter feeding identity}
\label{sec:twoatom}

Let \(N_e=\sum_j\tau_+^{(j)}\tau_-^{(j)}\) and \(Z=\langle a^\dagger S_-\rangle\).
The exact moment equations
\(\frac{d}{dt}\langle\nop\rangle=-A(\langle\nop\rangle-\nu)-g\operatorname{Im}Z\) and
\(\frac{d}{dt}\langle N_e\rangle=B(Ns-\langle N_e\rangle)+g\operatorname{Im}Z\)
give in stationarity
\begin{equation}
A(\langle\nop\rangle-\nu)=B(Ns-\langle N_e\rangle)=-g\operatorname{Im}Z ,
\label{eq:flux}
\end{equation}
and the polarization equation
\(\dot Z=-(A/2+C+i\Delta)Z-\tfrac{ig}{2}[\langle S_+S_-\rangle+\langle\nop S_z\rangle]\),
in which \(-A/2\) is independent of \(\nu\) and \(-C\) collects decay,
pump, and dephasing, yields
\begin{align}
A(\langle\nop\rangle-\nu)&=\mathcal R\big[\langle S_+S_-\rangle+\langle\nop S_z\rangle\big],\notag\\
\mathcal R&=\frac{g^2(A/2+C)}{2[(A/2+C)^2+\Delta^2]} .
\label{eq:flux2}
\end{align}
\(\mathcal R\) is a Lorentzian response coefficient of the exact identity;
it reduces to the familiar effective emission/absorption rate only in the
eliminated-atom limit of Sec.~\ref{sec:discussion}.  For two atoms, with
\(\mathcal C_{12}=\langle\tau_+^{(1)}\tau_-^{(2)}\rangle\), permutation
invariance gives
\(\mathcal C_{12}=\langle\tau_+^{(2)}\tau_-^{(1)}\rangle=\mathcal C_{12}^*\),
so \(\mathcal C_{12}\) is real, \(\langle S_+S_-\rangle=\langle N_e\rangle+2\mathcal C_{12}\),
and
\begin{equation}
A(\langle\nop\rangle-\nu)=\mathcal R\langle N_e\rangle
+2\mathcal R\,\mathcal C_{12}+\mathcal R\langle\nop S_z\rangle .
\label{eq:coherenceflux}
\end{equation}
The decomposition is exact; ``population,'' ``mutual coherence,'' and
``saturation'' are interpretive names for the three terms.
Table~\ref{tab:balance} evaluates them on the exact stationary state.

\begin{table}[tbp]
\caption{Two-atom photon balance, Eq.~(\ref{eq:coherenceflux}), at the
baseline parameters; \(\eta_C=2\mathcal C_{12}/\langle N_e\rangle\) is the
mutual-coherence term as a fraction of the population term.}
\label{tab:balance}
\begin{ruledtabular}
\begin{tabular}{cccccc}
\(s\) & \(A\langle\nop\rangle\) & \(\mathcal R\langle N_e\rangle\) &
\(2\mathcal R\,\mathcal C_{12}\) & \(\mathcal R\langle\nop S_z\rangle\) & \(\eta_C\)\\
\hline
0.05 & 0.0166 & 0.0543 & \(-0.0153\) & \(-0.0224\) & \(-28.1\%\)\\
0.20 & 0.0743 & 0.2092 & \(-0.0496\) & \(-0.0854\) & \(-23.7\%\)\\
0.40 & 0.1708 & 0.3957 & \(-0.0697\) & \(-0.1552\) & \(-17.6\%\)\\
0.60 & 0.2922 & 0.5570 & \(-0.0621\) & \(-0.2027\) & \(-11.1\%\)\\
0.80 & 0.4408 & 0.6906 & \(-0.0289\) & \(-0.2209\) & \(-4.2\%\)\\
0.90 & 0.5260 & 0.7464 & \(-0.0034\) & \(-0.2170\) & \(-0.5\%\)\\
\end{tabular}
\end{ruledtabular}
\end{table}

The mutual-coherence term is not a small correction: at \(s=0.05\) it
removes 28\% of the population term, and together with the saturation term
reduces the feeding rate to 31\% of the population-only estimate.  Its
sign has an exact representation-theoretic meaning.  In the basis
\(|1,0\rangle=(|eg\rangle+|ge\rangle)/\sqrt2\),
\(|0,0\rangle=(|eg\rangle-|ge\rangle)/\sqrt2\), the operator
\(\tau_+^{(1)}\tau_-^{(2)}\) acts as
\(|1,0\rangle\mapsto\tfrac12(|1,0\rangle+|0,0\rangle)\),
\(|0,0\rangle\mapsto-\tfrac12(|1,0\rangle+|0,0\rangle)\), and the
off-diagonal elements \(\langle1,0|\rho|0,0\rangle\) vanish because
the permutation-invariant \(\rho\) commutes with the swap, under which
\(|1,0\rangle\) is even and \(|0,0\rangle\) odd, so
\begin{equation}
2\,\mathcal C_{12}=p_{1,0}-p_{0,0},
\label{eq:coherencesinglet}
\end{equation}
with \(p_{1,0},p_{0,0}\) the triplet-\(m{=}0\) and singlet populations.  A
negative coherence term means that the dark singlet is more populated
than the bright \(m=0\) triplet state: the pair is subradiant relative to
two independent atoms.  Since the singlet is annihilated by \(S_\pm\)
(Sec.~\ref{sec:sectors}), only local pump and decay move population into
and out of it, and in this regime they trap more population there than
the cavity removes from the triplet.

The sign is not universal, and this is the main lesson of the identity:
cavity-mediated pair coherence is not intrinsically superradiant.
Section~\ref{sec:benchmark} shows that at \(s=0.9\) the pair coherence
becomes positive for \(N\ge3\) even in this bad-cavity regime
(Fig.~\ref{fig:Nscaling}c), and in the good-cavity regime of
Sec.~\ref{sec:lasing} it is positive already for \(N=2\) above the
lasing crossover (\(\mathcal C_{12}=+0.048\) at \(s=0.9\)), where the
cavity field builds up an in-phase, superradiant coherence between the
atoms.  Equation~(\ref{eq:coherenceflux}) is the exact diagnostic that
decides, for given parameters, which regime one is in.

\subsection{Weak-pump hierarchy of connected correlations}
\label{sec:hierarchy}

For \(N=3\), an atomic monomial with \(r\) raising and \(l\) lowering
operators has \(U(1)\) charge \(r-l\), and its stationary expectation
vanishes unless \(r=l\).  The natural charge-neutral connected
three-atom correlator is
\begin{equation}
\Gamma_3\equiv C_{3z}=\langle\tau_+^{(1)}\tau_-^{(2)}\sigma_z^{(3)}\rangle
-\langle\tau_+^{(1)}\tau_-^{(2)}\rangle\langle\sigma_z^{(3)}\rangle ,
\label{eq:C3z}
\end{equation}
and with \(n_j=\tau_+^{(j)}\tau_-^{(j)}\) the third excitation cumulant is
\begin{equation}
K_3=\langle n_1n_2n_3\rangle-3\langle n_1n_2\rangle\langle n_1\rangle+2\langle n_1\rangle^3 ,
\label{eq:K3}
\end{equation}
the residual of the second-order cumulant closure
\cite{ZhouZengXuYou2006}.  We write
\(C_2^{(N)}=\langle\tau_+^{(1)}\tau_-^{(2)}\rangle\) for the pair coherence
as a marginal of the \(N\)-atom state (at \(s=0.3\): \(-0.044\) for
\(N=2\), \(-0.032\) for \(N=3\)).  Figure~\ref{fig:C3zK3} shows
\(\Gamma_3\) and \(K_3\) at \(s=0.6\) as functions of coupling and
detuning: both are interaction generated and resonance sensitive,
\(\Gamma_3\) is negligible below \(g/A\approx0.4\), reaches
\(5.0\times10^{-3}\) at the baseline coupling (about one third of the
factorized value \(|C_2^{(3)}\langle\sigma_z\rangle|\)), and changes sign
at \(|\Delta|/A\approx0.65\); \(K_3\) is positive and an order of magnitude
smaller.

\begin{figure}[tbp]
\includegraphics[width=\linewidth]{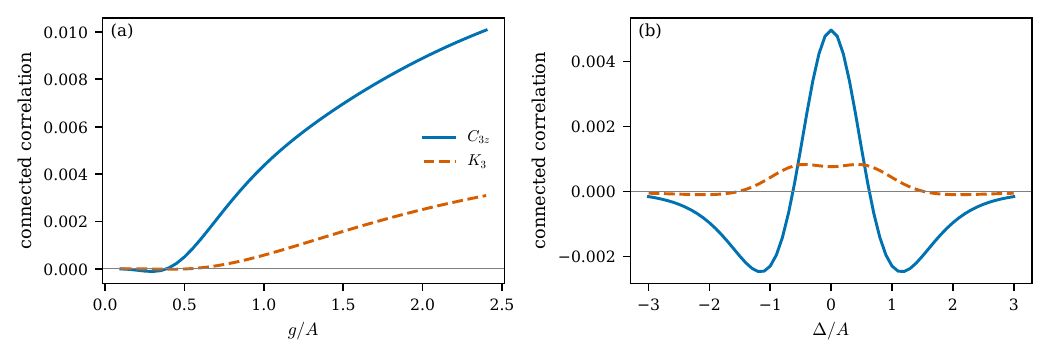}
\caption{Connected three-atom correlations \(\Gamma_3=C_{3z}\) (solid) and
\(K_3\) (dashed) at \(s=0.6\), \(A=1\), \(B=0.7\), \(C=B/2\), \(\nu=0\):
(a) versus coupling at \(\Delta=0\); (b) versus detuning at \(g=1.1\).}
\label{fig:C3zK3}
\end{figure}

The pump order at which these structures first appear follows from an
excitation filtration.  Throughout this subsection the cavity is at zero
temperature, \(\nu=0\), as in all baseline calculations; for \(\nu>0\)
the term \(A\nu\mathcal D[a^\dagger]\) in \(\mathcal L_0\) raises the
excitation and the filtration below does not apply.  The Liouvillian is
linear in \(s\),
\(\mathcal L(s)=\mathcal L_0+s\mathcal L_1\) with
\(\mathcal L_1=B\sum_j(\mathcal D[\tau_+^{(j)}]-\mathcal D[\tau_-^{(j)}])\),
and at \(s=0\) the stationary state is \(\rho^{(0)}=|0\rangle\langle0|\otimes|g\cdots g\rangle\langle g\cdots g|\).
In any finite Fock truncation \(\mathcal L_0\) has a simple zero eigenvalue,
so the stationary state has the regular expansion
\(\rho_{\rm ss}(s)=\sum_rs^r\rho^{(r)}\) with
\(\mathcal L_0\rho^{(r)}=-\mathcal L_1\rho^{(r-1)}\), \(\operatorname{tr}\rho^{(r)}=0\)
for \(r\ge1\); all statements below hold at the level of this series and
are stable under increase of the truncation.  Let
\(\mathcal E=\nop+\sum_jn_j\) be the total excitation operator, with
eigenstates \(|i\rangle\) and eigenvalues \(E_i\), and let
\(V_r=\operatorname{span}\{|i\rangle\langle j|:E_i\le r,\ E_j\le r\}\).

\begin{lemma}[Weak-pump excitation filtration, \(\nu=0\)]
\label{lem:filtration}
For \(\nu=0\): \(\mathcal L_0V_r\subseteq V_r\), \(\mathcal L_1V_r\subseteq V_{r+1}\), and
\(\rho^{(r)}\in V_r\) for all \(r\ge0\).
\end{lemma}
\begin{proof}
The Tavis--Cummings Hamiltonian and \(\nop\), \(S_z\) commute with
\(\mathcal E\), so the commutator term maps \(V_r\) to itself.  For each
jump operator \(L\in\{a,\tau_-^{(j)},\sigma_z^{(j)}\}\) of \(\mathcal L_0\),
\(L\) either lowers or preserves \(\mathcal E\), so \(L|i\rangle\langle j|L^\dagger\)
and \(\{L^\dagger L,|i\rangle\langle j|\}\) stay in \(V_r\); hence
\(\mathcal L_0V_r\subseteq V_r\).  In \(\mathcal L_1\), the sandwich term
\(\tau_+^{(j)}|i\rangle\langle j|\tau_-^{(j)}\) raises both \(E_i\) and \(E_j\)
by at most one and the anticommutator terms preserve them, so
\(\mathcal L_1V_r\subseteq V_{r+1}\).  Finally, \(\rho^{(0)}\in V_0\), and if
\(\rho^{(r-1)}\in V_{r-1}\) then \(\mathcal L_1\rho^{(r-1)}\in V_r\); since
\(V_r\) is a finite-dimensional invariant subspace of \(\mathcal L_0\)
containing the stationary state and on which zero is a simple eigenvalue,
the traceless solution of \(\mathcal L_0\rho^{(r)}=-\mathcal L_1\rho^{(r-1)}\)
lies in \(V_r\).
\end{proof}

The filtration controls the Taylor coefficients of the perturbation
series.  The pump orders themselves can be obtained directly for the
stationary state, without reference to a series, from a stationary bound
on the factorial moments of the total excitation, which also holds
uniformly in the Fock cutoff and therefore transfers to \(\rho_{\rm ss}\)
by Proposition~\ref{prop:fock}.

\begin{lemma}[Stationary excitation moments, \(\nu=0\)]
\label{lem:factorial}
Let \(\nu=0\), \((\Exc)_r=\Exc(\Exc-1)\cdots(\Exc-r+1)\), and
\(c_s=\min\{A,B(1-s)\}\).  For every Fock cutoff and for \(\rho_{\rm ss}\)
under (U),
\begin{equation}
\langle(\Exc)_r\rangle\le\Big(\frac{BsN}{c_s}\Big)^{r},\qquad r\ge0 ,
\label{eq:factorialbound}
\end{equation}
and consequently \(\Pr(\Exc\ge k)\le\langle(\Exc)_k\rangle/k!=O(s^k)\) and
\(|\langle O\rangle|\le\|O\|\,\Pr(\Exc\ge k)=O(s^k)\) for every operator
\(O=P_{\Exc\ge k}OP_{\Exc\ge k}\) supported on \(\Exc\ge k\).
\end{lemma}
\begin{proof}
As in Lemma~\ref{lem:expmoments}, \(H\) does not contribute to
\(\mathcal L^\dagger f(\Exc)\), the up-jumps \(\tau_+^{(j)}\) occur at
total rate \(b=Bs(N-N_e)\le BsN\) and raise \(\Exc\) by one, the
down-jumps \(a\), \(\tau_-^{(j)}\) occur at total rate
\(d=A\nop+B(1-s)N_e\ge c_s\Exc\) and lower it by one, so
\(\mathcal L^\dagger(\Exc)_r=b[(\Exc+1)_r-(\Exc)_r]+d[(\Exc-1)_r-(\Exc)_r]
=rb(\Exc)_{r-1}-rd(\Exc-1)_{r-1}\).  For \(r\ge1\): on \(\Exc\ge1\) one has \((\Exc-1)_{r-1}\ge0\) (it
vanishes for \(1\le\Exc\le r-1\)) and \(d\ge c_s\Exc\), while at
\(\Exc=0\) both \(d\) and \(\Exc\) vanish, so
\(d(\Exc-1)_{r-1}\ge c_s\Exc(\Exc-1)_{r-1}=c_s(\Exc)_r\) on the entire
spectrum; stationarity then gives
\(c_s\langle(\Exc)_r\rangle\le\langle d(\Exc-1)_{r-1}\rangle=\langle b(\Exc)_{r-1}\rangle\le BsN\langle(\Exc)_{r-1}\rangle\),
and induction from the base \(\langle(\Exc)_0\rangle=1\).  Since
\((\Exc)_k/k!=\binom{\Exc}{k}\ge1\) on \(\Exc\ge k\), the tail bound
follows, and \(|\Tr(\rho P_{\ge k}OP_{\ge k})|\le\|O\|\Tr(\rho P_{\ge k})\).
\end{proof}

\begin{proposition}[Pump order of excitation cumulants, \(\nu=0\)]
\label{prop:Km-scaling}
For \(\nu=0\) and any \(k\) distinct atoms, \(\langle n_{i_1}\cdots n_{i_k}\rangle=O(s^k)\),
and for \(2\le m\le N\) the joint cumulant \(K_m=\kappa(n_1,\ldots,n_m)\)
satisfies \(K_m=O(s^m)\) as \(s\to0\).
\end{proposition}
\begin{proof}
\(n_{i_1}\cdots n_{i_k}\) is a projector supported on \(\Exc\ge k\), so
the moment statement is Lemma~\ref{lem:factorial} (equivalently, by
Lemma~\ref{lem:filtration}, the diagonal elements
\(\langle i|\rho^{(r)}|i\rangle\) with the \(k\) atoms excited have
\(E_i\ge k\) and vanish for \(r<k\)).  By the moment--cumulant relation,
\(K_m\) is a finite sum over set partitions of \(\{1,\ldots,m\}\) of
products of such moments; a partition with block sizes \(k_1,\ldots,k_p\),
\(\sum_jk_j=m\), contributes \(O(s^{k_1})\cdots O(s^{k_p})=O(s^m)\).  Every
term is individually \(O(s^m)\).
\end{proof}

The same filtration controls the connected coherences
\(\Gamma_m=\kappa(\tau_+^{(1)}\tau_-^{(2)},\sigma_z^{(3)},\ldots,\sigma_z^{(m)})\),
\(\Gamma_3=C_{3z}\), once one uses that a joint cumulant of two or more
variables is unchanged when any variable is shifted by a constant: with
\(\sigma_z^{(j)}=2n_j-1\),
\(\Gamma_m=2^{m-2}\kappa(\tau_+^{(1)}\tau_-^{(2)},n_3,\ldots,n_m)\).

\begin{proposition}[Pump order of connected coherences, \(\nu=0\)]
\label{prop:Gm-scaling}
For \(\nu=0\), \(C_2^{(N)}=O(s)\), and for \(3\le m\le N\),
\(\Gamma_m=O(s^{m-1})\) as \(s\to0\).
\end{proposition}
\begin{proof}
The operator \(O=\tau_+^{(1)}\tau_-^{(2)}n_{j_1}\cdots n_{j_b}\)
(\(j_1,\ldots,j_b\) distinct and different from 1, 2) has norm one and
maps states with atom 2 and atoms \(j_1,\ldots,j_b\) excited to states
with atom 1 and atoms \(j_1,\ldots,j_b\) excited, so
\(O=P_{\Exc\ge b+1}OP_{\Exc\ge b+1}\) and \(\langle O\rangle=O(s^{b+1})\)
by Lemma~\ref{lem:factorial} (equivalently, the contributing elements
\(\langle j|\rho^{(r)}|i\rangle\) have \(E_i,E_j\ge b+1\) and vanish for
\(r<b+1\) by Lemma~\ref{lem:filtration}).  For \(b=0\) this is the
statement for \(C_2^{(N)}\).
In the partition expansion of
\(\kappa(\tau_+^{(1)}\tau_-^{(2)},n_3,\ldots,n_m)\) the block containing
\(\tau_+^{(1)}\tau_-^{(2)}\) together with \(b\) excitation variables
contributes \(O(s^{b+1})\) and a block of \(k\) excitation variables
\(O(s^k)\) (Proposition~\ref{prop:Km-scaling}), so every product is
\(O(s^{1+(m-2)})=O(s^{m-1})\).
\end{proof}

Propositions~\ref{prop:Km-scaling} and \ref{prop:Gm-scaling} are upper
bounds; whether they are attained --- whether the leading coefficients
\(\gamma_m^{(N)}\) and \(\kappa_m^{(N)}\) in
\(\Gamma_m=\gamma_m^{(N)}s^{m-1}+O(s^m)\) and
\(K_m=\kappa_m^{(N)}s^m+O(s^{m+1})\) are nonzero --- is not fixed by the
bounds, and we settle it numerically, from the Taylor coefficients of
the stationary state rather than from logarithmic slopes at small
\(s\).  The
Liouvillian is linear in \(s\), and so is the block-Jacobi matrix
\(\mathcal M(s)=\mathcal M_0+s\mathcal M_1\) of
Eq.~(\ref{eq:stationary}) in the \(s\)-dependent basis
\(\Phi_{n\alpha}(s)\), because the matrices of
Eqs.~(\ref{eq:SLplus})--(\ref{eq:SRminus}) are affine in \(s\).  The
stationary coefficient vector therefore has the regular expansion
\(X(s)=\sum_rs^rX^{(r)}\), \(\mathcal M_0X^{(r)}=-\mathcal M_1X^{(r-1)}\),
with \(X^{(0)}=\Phi_{0,(N,0,0,0)}\) and the trace component of
\(X^{(r)}\) zero for \(r\ge1\).  Moreover \(X^{(r)}_n=0\) for \(n>r\).
Indeed, an operator in \(V_r\) has vanishing field-dual coefficients
\(\Tr(\check\rho_n^{(k)}\,\cdot\,)\) for all \(n>r\), because the
monomials \(\check\rho_n^{(k)}\) of Sec.~\ref{sec:branch} contain at
least \(n\) annihilation operators on one side and so vanish on Fock
levels below \(n\); writing
\(\rho^{(r)}=\sum_{p+q=r}\sum_\alpha X^{(p)}_{n\alpha}\rho_n^{(k_\alpha)}\otimes R^{[q]}_{\mv_\alpha}\),
with \(R_{\mv}(s)=\sum_qs^qR^{[q]}_{\mv}\) polynomial in \(s\), the terms
with \(q\ge1\) have \(n\le r-q\) by the induction hypothesis, so for
\(n>r\) the operator
\(\sum_\alpha X^{(r)}_{n\alpha}\rho_n^{(k_\alpha)}\otimes R_{\mv_\alpha}(0)\)
has zero field-dual coefficients, and linear independence of the
\(R_{\mv}(0)\) gives \(X^{(r)}_n=0\).  (For the same reason
\(\rho^{(r)}\) is independent of the Fock cutoff once
\(n_{\rm c}>r\).)  The order-\(r\) coefficient is therefore the solution
of a \emph{finite} linear system, free of radial truncation error once
\(n_{\max}\ge r\); we evaluate these systems numerically, so
floating-point error remains, and obtain the moments
\(\sum_\alpha X_{0\alpha}(s)\Tr[R_{\mv_\alpha}(s)O]\) and their cumulants
by truncated power-series arithmetic.  Table~\ref{tab:series}
gives the result for \(N=3,\ldots,8\) and \(3\le m\le N\), through
\(m=8\).  In every case the coefficients below the proved order vanish
to round-off (largest \(4\times10^{-11}\), at \(N=8\); below \(10^{-14}\) for
\(N\le5\)), the computed \(X^{(r)}_n\) vanish for \(n>r\) as the lemma
requires, and every leading coefficient is nonzero: the bounds of both
propositions are attained for all \(m\le N\le8\) at the baseline
parameters.  Because the trace-constrained solves and the cumulant
subtractions are not perfectly conditioned (the second-smallest singular
value of \(\mathcal M_0\) is \(2\times10^{-5}\) at \(N=8\)), the
coefficients were re-evaluated in two independent ways
(Appendix~\ref{app:series-checks}): in 40-digit arithmetic by iterative
refinement against the exactly represented systems for \(N=6\) and
\(N=8\), which reproduces the double-precision leading coefficients to
\(5\times10^{-12}\) (\(N=6\)) and \(1\times10^{-8}\) (\(N=8\); median \(3\times10^{-10}\)) and gives the coefficients below the leading pump order as zero to
\(7\times10^{-37}\) or better; and by perturbation theory in the full, non-symmetry-reduced
Fock\(\otimes\)atom space for \(N=4,5\), which agrees to \(2\times10^{-12}\).  The coefficients are neither monotone nor of fixed sign in
\(N\) --- \(\gamma_3^{(N)}\) changes sign between \(N=4\) and \(5\), and
\(\kappa_5^{(7)}\) is anomalously small --- so a nonvanishing theorem for
general \(m,N\) would have to rest on more than the filtration, and we
do not claim one.  The series are confirmed by the full solution at
finite pump: at \(s=0.02\) and \(N=6\) the exact \(\Gamma_6\) and \(K_6\)
are \(-2.7385\times10^{-11}\) and \(2.1403\times10^{-13}\), against
\(-2.7386\times10^{-11}\) and \(2.1404\times10^{-13}\) from the series
truncated at order \(N+2\); at \(N=3\) the coefficients
\(C_2^{(3)}=-0.164380\,s\), \(\Gamma_3=0.043552\,s^2\),
\(K_3=0.032190\,s^3\) reproduce the values fitted from the exact
solution over \(10^{-3}\le s\le10^{-2}\) in Fig.~\ref{fig:hierarchy},
and the next coefficients (\(\Gamma_3^{(4)}=2.6\times10^{-3}s^2+4.9\times10^{-2}s^3+\ldots\))
explain the departures of the finite-\(s\) logarithmic slopes in that
figure (\(2.06\) and \(4.06\) for \(\Gamma_3\) and \(K_4\) at \(N=4\))
from the exact exponents.  Each additional atom entering a connected
correlation requires at least one further power of the pump, and at the
baseline parameters exactly one --- a hierarchy that
low-order cumulant treatments impose by truncation rather than derive.

\begin{table*}[tbp]
\caption{Leading coefficients of the stationary weak-pump Taylor
expansion for the connected coherences
\(\Gamma_m=\gamma_m^{(N)}s^{m-1}+O(s^m)\) (left) and the excitation
cumulants \(K_m=\kappa_m^{(N)}s^m+O(s^{m+1})\) (right) at the baseline
parameters, \(\nu=0\), in units of \(10^{-3}\), evaluated numerically
from finite linear systems that carry no radial truncation error
(Sec.~\ref{sec:hierarchy}).  All
coefficients below the proved orders (Propositions~\ref{prop:Km-scaling}
and \ref{prop:Gm-scaling}) vanish to round-off: their largest modulus is
below \(10^{-14}\) for \(N\le5\), \(10^{-13}\) for \(N=6\), \(2\times10^{-12}\)
for \(N=7\), and \(4\times10^{-11}\) for \(N=8\).  For \(N=3\) the
coefficients are, to the digits shown, those fitted from the exact
solution in Fig.~\ref{fig:hierarchy}.}
\label{tab:series}
\begin{ruledtabular}
\begin{tabular}{c|rrrrrr|rrrrrr}
 & \multicolumn{6}{c|}{\(10^{3}\gamma_m^{(N)}\)} & \multicolumn{6}{c}{\(10^{3}\kappa_m^{(N)}\)}\\
\(m\) & \(N=3\) & \(N=4\) & \(N=5\) & \(N=6\) & \(N=7\) & \(N=8\) & \(N=3\) & \(N=4\) & \(N=5\) & \(N=6\) & \(N=7\) & \(N=8\)\\
\hline
3 & 43.55 & 2.624 & \(-7.125\) & \(-9.142\) & \(-8.958\) & \(-8.163\) & 32.19 & 12.00 & \(-0.8959\) & \(-4.246\) & \(-4.591\) & \(-4.112\)\\
4 & & 27.57 & 28.60 & 16.38 & 8.856 & 4.773 & & \(-1.294\) & 9.898 & 7.843 & 4.002 & 1.794\\
5 & & & \(-45.30\) & \(-9.752\) & 4.039 & 5.632 & & & \(-12.75\) & \(-6.162\) & \(-0.02544\) & 1.617\\
6 & & & & \(-12.14\) & \(-28.50\) & \(-16.18\) & & & & 3.100 & \(-3.744\) & \(-3.442\)\\
7 & & & & & 69.39 & 18.83 & & & & & 6.582 & 3.544\\
8 & & & & & & 9.261 & & & & & & \(-2.031\)\\
\end{tabular}
\end{ruledtabular}
\end{table*}

\begin{figure}[tbp]
\includegraphics[width=\linewidth]{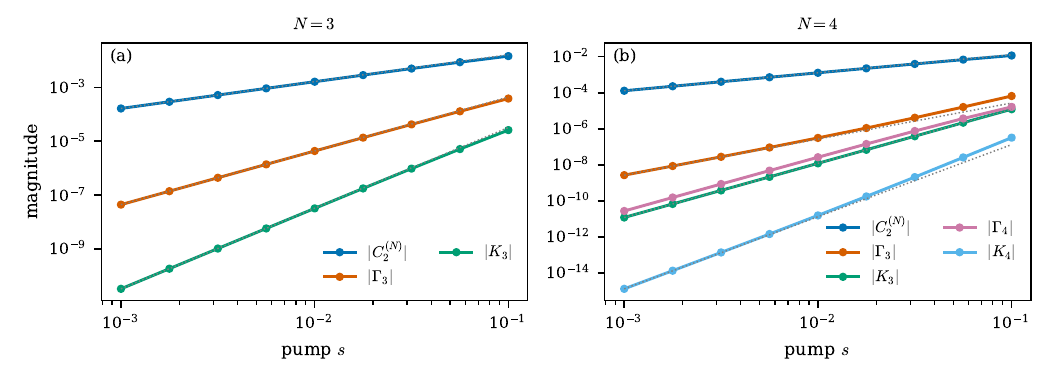}
\caption{Weak-pump hierarchy of connected correlations at the baseline
parameters for (a) \(N=3\) and (b) \(N=4\), from the exact solution at
finite pump.  Dotted lines are the power
laws \(s^1,s^2,s^3,s^4\) anchored at the smallest pump; the exact
exponents and leading coefficients are those of Table~\ref{tab:series}.}
\label{fig:hierarchy}
\end{figure}

\subsection{Collective-spin sectors and their parity}
\label{sec:sectors}

With \(\mathbf J=\tfrac12\sum_j\boldsymbol\sigma^{(j)}\), the normalized
ladders give \(S_\pm=J_\pm\), which commute with \(\mathbf J^2\) and are
block diagonal in the total spin \(J\) \cite{Dicke1954}.

\begin{proposition}[Parity of structural cavity darkness]
\label{prop:parity}
The collective coupling annihilates the minimal-spin sector,
\(S_\pm P_{J_{\min}}=0\), if and only if \(N\) is even.
\end{proposition}
\begin{proof}
\(J_{\min}=0\) for even \(N\) and \(\tfrac12\) for odd \(N\).  A \(J=0\)
multiplet is one-dimensional, so \(J_\pm\) annihilate it; for
\(J=\tfrac12\), \(J_+|\tfrac12,-\tfrac12\rangle=|\tfrac12,\tfrac12\rangle\neq0\).
\end{proof}

This is a kinematic statement about the coupling, not about stationary
feeding: for odd \(N\) the minimal sector is not symmetry-protected, but
whether and how strongly it feeds the cavity is a dynamical question
(trivially, it does not for \(g=0\)).  The standard-form diagnostic is the
sector-resolved feeding rate
\begin{equation}
\Phi_J=-g\,\frac{\operatorname{Im}\langle a^\dagger S_-P_J\rangle}{\langle P_J\rangle},
\label{eq:PhiJ}
\end{equation}
the rate at which sector \(J\) feeds the cavity per unit population, the
sector-conditional version of Eq.~(\ref{eq:flux}), obtained from \(X_0\)
by Eq.~(\ref{eq:joint-observable}) with \(O=S_-P_J\) and in agreement with
the brute-force value to \(10^{-16}\).  For even \(N\),
\(\Phi_{J_{\min}}=0\) exactly.  Table~\ref{tab:sectors} and
Fig.~\ref{fig:sectorN3} show that for \(N=3\) in the baseline regime the
\(J=1/2\) sector is cavity active, \(\Phi_{1/2}>0\), but subradiant:
\(\Phi_{3/2}/\Phi_{1/2}\) grows from 1.5 at weak pump to 4.8 near
\(s=0.85\), while the sector traps up to 55\% of the population.
Section~\ref{sec:parityN} follows this to \(N=8\).

\begin{table}[tbp]
\caption{Sector-resolved structure at the baseline parameters: population
\(p_J\) and conditional photon number \(\bar n_J=\langle\nop P_J\rangle/p_J\)
of the minimal sector, and feeding rates of the minimal and maximal
sectors (\(J=1\) for \(N=2\), \(J=3/2\) for \(N=3\)).}
\label{tab:sectors}
\begin{ruledtabular}
\begin{tabular}{ccccccc}
\(N\) & \(s\) & \(J_{\min}\) & \(p_{J_{\min}}\) & \(\bar n_{J_{\min}}\) & \(\Phi_{J_{\min}}/A\) & \(\Phi_{J_{\max}}/A\)\\
\hline
2 & 0.05 & 0 & 4.8\% & 0.0069 & 0 & 0.0174\\
2 & 0.30 & 0 & 21.5\% & 0.0532 & 0 & 0.1524\\
2 & 0.60 & 0 & 28.3\% & 0.1495 & 0 & 0.4074\\
2 & 0.90 & 0 & 24.1\% & 0.3399 & 0 & 0.6934\\
3 & 0.05 & 1/2 & 9.5\% & 0.0146 & 0.0118 & 0.0191\\
3 & 0.30 & 1/2 & 42.4\% & 0.1072 & 0.0741 & 0.2035\\
3 & 0.60 & 1/2 & 54.8\% & 0.2890 & 0.1583 & 0.7003\\
3 & 0.90 & 1/2 & 47.3\% & 0.6253 & 0.2597 & 1.2491\\
\end{tabular}
\end{ruledtabular}
\end{table}

\begin{figure}[tbp]
\includegraphics[width=\linewidth]{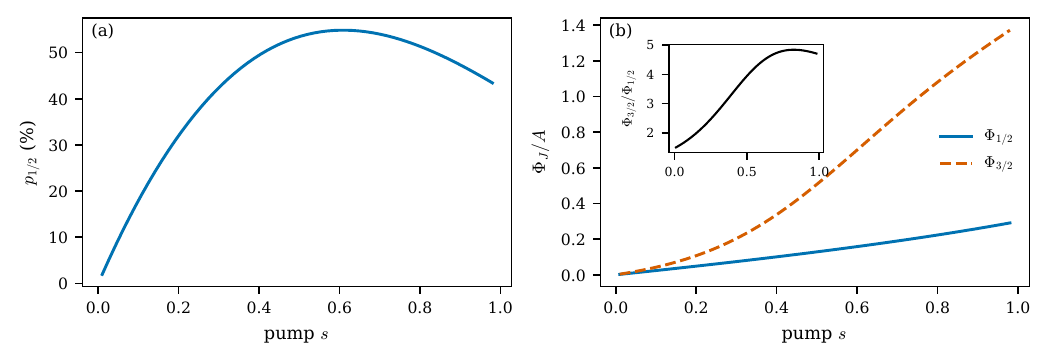}
\caption{Sector-resolved structure of the \(N=3\) stationary state at the
baseline parameters.  (a) Population of the \(J=1/2\) sector, peaking at
55\% near \(s=0.6\).  (b) Feeding rates \(\Phi_{1/2}\) (solid) and
\(\Phi_{3/2}\) (dashed); inset: their ratio.}
\label{fig:sectorN3}
\end{figure}

\subsection{Irreducible three-body information}
\label{sec:maxent}

A complementary question is whether the three-atom state is determined by
its one- and two-body marginals.  With \(\rho_3\) the exact reduced atomic
state and \(\rho_3^{(2)}\) the maximum-entropy state with the same
permutation-symmetric one- and two-body marginals (nine constraint
operators; Appendix~\ref{app:maxent}), the irreducible three-body
information \(I_3^{\rm irr}=S(\rho_3^{(2)})-S(\rho_3)\ge0\)
\cite{LindenPopescuWootters2002,Zhou2008,Zhou2009} is
\(1.08,\,2.89,\,3.48,\,2.82\times10^{-4}\) nats at \(s=0.3,0.5,0.7,0.9\),
resolved by eight orders of magnitude above the constraint residuals
(and unchanged to \(10^{-12}\) when the Newton tolerance is tightened
from \(10^{-9}\) to \(10^{-13}\)), numerically consistent with a
quadratic onset over the investigated range of small \(s\).  It is
small, but its meaning connects directly to Sec.~\ref{sec:benchmark}: a
positive \(I_3^{\rm irr}\) shows that the exact three-atom state differs
from the maximum-entropy state compatible with its one- and two-body
marginals, i.e.\ that those marginals do not determine the state, and
the maximum-entropy reconstruction constrained by them underestimates
\(\Gamma_3\) by a factor 4 at \(s=0.3\) and 12.5 at \(s=0.5\), and misses
it by three orders of magnitude or with the wrong sign at \(s=0.7\) and
\(0.9\) (Table~\ref{tab:maxent} in Appendix~\ref{app:maxent}).  This
says nothing about reconstruction procedures supplied with additional
model information.

\section{Scaling with \(N\) and the second-order cumulant closure}
\label{sec:benchmark}

\subsection{Results to \(N=16\)}
\label{sec:Nscaling}

Figure~\ref{fig:Nscaling} shows \(\langle\nop\rangle\), \(g^{(2)}(0)\), and
\(C_2^{(N)}\) for \(N=1,\ldots,16\) at three pump values in the baseline
regime.  At weak pump, \(s=0.3\), the photon number saturates at
\(\langle\nop\rangle\approx0.27\) while \(g^{(2)}(0)\) rises to 2.15, slightly
above thermal; at \(s=0.6\) it saturates at 1.61; at \(s=0.9\) the photon
number grows linearly, 0.26 per atom, and \(g^{(2)}(0)\) passes through a
maximum of 1.256 at \(N=8\).  None of these curves shows even/odd
structure.  The pair coherence is negative and decreasing in magnitude
roughly as \(1/N\) at \(s=0.3\) and \(0.6\), but at \(s=0.9\) it changes
sign between \(N=2\) and \(N=3\) and settles at \(+0.009\).

\begin{figure*}[tbp]
\includegraphics[width=\textwidth]{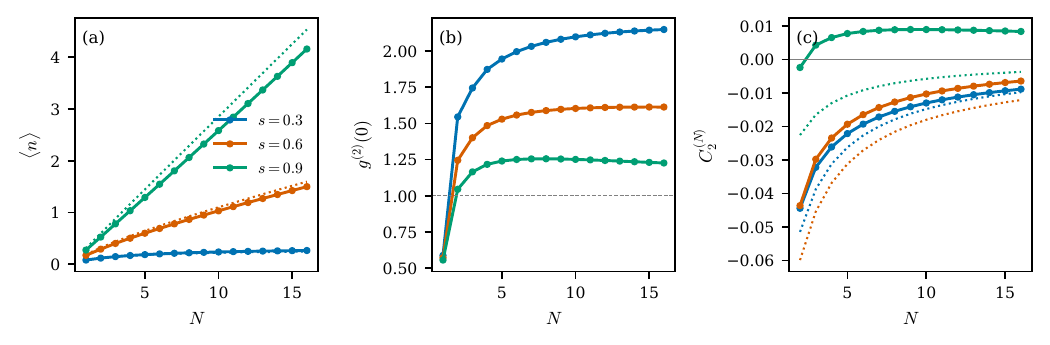}
\caption{Scaling with atom number at the baseline parameters for
\(s=0.3,0.6,0.9\): (a) mean photon number, (b) \(g^{(2)}(0)\), (c) pair
coherence \(C_2^{(N)}\).  Symbols: exact continued fraction; dotted: the
second-order cumulant closure, Eq.~(\ref{eq:closure}).  The continued-fraction results for \(N\le8\) are reproducible to at
least 11 significant digits; for \(N\ge12\) at \(s=0.9\) the plain
double-precision recursion exhibits a cutoff and terminal-condition
sensitivity corresponding to roughly three significant digits in (a) and
(b) and two to three in (c)
(Sec.~\ref{sec:precision}, Table~\ref{tab:Nconditioning}); the observed
numerical variations are smaller than the plotting resolution.}
\label{fig:Nscaling}
\end{figure*}

\subsection{Parity of the minimal sector for \(N\le8\)}
\label{sec:parityN}

Figure~\ref{fig:parity} extends the sector analysis to \(N=8\).  The
feeding rate of the minimal sector vanishes identically for even \(N\)
(to \(10^{-17}\)) and is \(0.06\)--\(0.16\,A\) for odd \(N\), nearly
independent of \(N\), while the maximal-\(J\) rate grows with \(N\); the
trapped population also alternates, \(42\%,9\%,22\%,5\%,13\%,3\%\) for
\(N=3,\ldots,8\) at \(s=0.3\), each odd-\(N\) value four to five times the
following even-\(N\) value (the multiplicities of \(J_{\min}\) coincide for
\(N=2m+1\) and \(2m+2\), so this is a population effect, not a counting
artifact).  Yet Fig.~\ref{fig:Nscaling} shows no trace of it in
\(\langle\nop\rangle\) or \(g^{(2)}(0)\).  The parity effect is thus not a
global even/odd modulation of the laser output: it is a redistribution of
stationary weight and radiative efficiency among collective-spin sectors,
and ordinary field observables average over sectors --- the minimal one
contributes little either way, being dark for even \(N\) and several times
subradiant for odd \(N\), while the many intermediate-\(J\) sectors
dominate the emission for \(N\ge4\).  The same representation-theoretic
even/odd distinction that produces parity-dependent photon statistics in
the spontaneous-emission model of Ref.~\cite{Auffeves2011} therefore
survives in the pumped laser, but its observable manifestation is
different.

\begin{figure}[tbp]
\includegraphics[width=\linewidth]{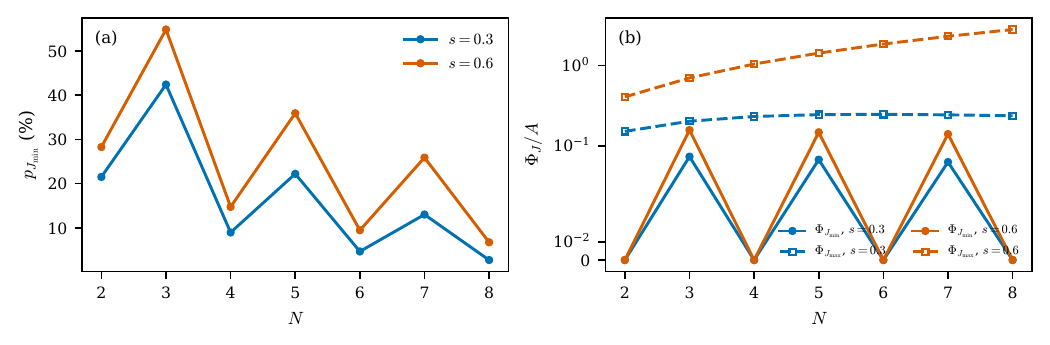}
\caption{Parity of the minimal collective-spin sector, \(N=2,\ldots,8\),
baseline parameters, \(s=0.3\) and \(0.6\): (a) population of the minimal
sector; (b) feeding rates of the minimal (filled circles) and maximal
(open squares) sectors, symmetric-log scale.}
\label{fig:parity}
\end{figure}

\subsection{Second-order cumulant closure}
\label{sec:closure-bench}

The standard approximate treatment of Eq.~(\ref{eq:master}) for several
emitters is the second-order cumulant closure
\cite{Gies2007,Leymann2014,MeiserHolland2010,KirtonKeeling2018}.  The
exact equations of motion for \(\langle\nop\rangle\),
\(p_e=\langle n_j\rangle\), \(Z_1=\langle a^\dagger\tau_-^{(j)}\rangle\), and
\(C=\langle\tau_+^{(i)}\tau_-^{(j)}\rangle\) (\(i\neq j\)) couple to the
third-order moments \(\langle a^\dagger a\,\sigma_z^{(j)}\rangle\) and
\(\langle a^\dagger\sigma_z^{(i)}\tau_-^{(j)}\rangle\), whose exact
equations couple to fourth order, and so on.  The closure sets the
third-order cumulants to zero; since the first moments of the charged
operators (\(a\), \(a^\dagger\), \(\tau_\pm\)) vanish by \(U(1)\)
symmetry while \(\langle\sigma_z\rangle\) is retained, this is the
factorization
\begin{equation}
\langle a^\dagger a\,\sigma_z^{(j)}\rangle\to\langle\nop\rangle\langle\sigma_z\rangle,\qquad
\langle a^\dagger\sigma_z^{(i)}\tau_-^{(j)}\rangle\to\langle\sigma_z\rangle Z_1 ,
\label{eq:factorization}
\end{equation}
i.e.\ the three-operator cumulants \(\kappa(a^\dagger,a,\sigma_z^{(j)})=0\) and
\(\kappa(a^\dagger,\sigma_z^{(i)},\tau_-^{(j)})=0\) (the remaining terms of
the cumulant expansion vanish with the first moments).  The closed stationary
equations are
\begin{align}
0&=-A(\langle\nop\rangle-\nu)-gN\operatorname{Im}Z_1,\notag\\
0&=B(s-p_e)+g\operatorname{Im}Z_1,\notag\\
0&=-(A/2+C_{\rm r}+i\Delta)Z_1\notag\\
&\quad-\tfrac{ig}{2}\big[p_e+(N-1)C+\langle\nop\rangle(2p_e-1)\big],\notag\\
0&=-2C_{\rm r}C-g(2p_e-1)\operatorname{Im}Z_1 ,
\label{eq:closure}
\end{align}
with \(C_{\rm r}\) the transverse rate; we obtain the physical root by
integrating the corresponding equations of motion from the vacuum.  The
hierarchy at this order contains no independent fourth-order field
moment, so \(g^{(2)}(0)\) is not determined without a further assumption;
the zero-mean Wick factorization
\(\langle a^{\dagger2}a^2\rangle=2\langle\nop\rangle^2\) gives
\(g^{(2)}(0)=2\) and misses the finite-\(N\) photon statistics of
Fig.~\ref{fig:Nscaling}b entirely.

The dotted curves in Figs.~\ref{fig:Nscaling}, \ref{fig:n_g2}, and
\ref{fig:goodcavity} compare the closure with the exact solution.  At weak
pump, \(s=0.3\), it deviates from the exact \(\langle\nop\rangle\) by up
to 8\% (above it for \(N\le3\), below it for \(N\ge4\)); at
\(s=0.6\) it overestimates by 9\% (\(N=3\)) to 7\% (\(N=16\)); at \(s=0.9\)
by 13\% (\(N=3\)) to 9\% (\(N=16\)); in the good-cavity regime by 3--18\% at
\(s=0.9\) and by up to 25\% near the lasing crossover.  It is thus
semi-quantitative for the photon number.  For the collective observable
it fails qualitatively: at \(s=0.9\) it predicts a pair coherence of the
wrong sign for all \(N\ge3\), \(C=-0.017\) at \(N=3\) against the exact
\(+0.004\), and it predicts \(C<0\) even where the exact
\(C_2^{(N)}\) is positive up to \(N=16\).  The exact solution identifies the
correlations omitted by the closure directly: Table~\ref{tab:omitted}
evaluates the two cumulants set to zero in Eq.~(\ref{eq:factorization}),
\(\delta_1=\langle a^\dagger a\,\sigma_z^{(j)}\rangle-\langle\nop\rangle\langle\sigma_z\rangle\)
and
\(\delta_2=\langle a^\dagger\sigma_z^{(i)}\tau_-^{(j)}\rangle-\langle\sigma_z\rangle Z_1\),
on the exact stationary state, from the \(n=1\), charge-zero and the
\(n=0\), charge-one components of \(X\) respectively
(Sec.~\ref{sec:observables}).  Both enter the polarization equation of
the closure, the first through the saturation term \(\langle\nop S_z\rangle\)
and the second through the equation for the pair coherence, whose
exact stationary form at resonance is
\begin{equation}
C_2^{(N)}=-\frac{g}{2C_{\rm r}}\Big[\langle\sigma_z\rangle\operatorname{Im}Z_1+\operatorname{Im}\delta_2\Big],
\label{eq:C2exact}
\end{equation}
the closure retaining the first term only (with its own self-consistent
\(\langle\sigma_z\rangle\) and \(Z_1\)).  Table~\ref{tab:omitted} lists
both terms with their signs.  Neither cumulant is small:
\(\delta_1\) is 18\% of the factorized value at weak pump and exceeds it
in magnitude from \(s=0.6\) on, where \(\langle\sigma_z\rangle\) changes
sign while \(\langle\nop\sigma_z\rangle\) does not --- the photon number
and the inversion are anticorrelated, which the factorization cannot
represent --- and \(\operatorname{Im}\delta_2\) opposes the factorized
term in every row, reducing \(|C_2|\) by 15--18\% at weak and
intermediate pump and in the good-cavity regime, and by 80\% at the
baseline \(s=0.9\).
Equation~(\ref{eq:C2exact}) also separates the direct from the
indirect error of the closure.  At the baseline \(s=0.9\) the exact
inversion is already positive, \(\langle\sigma_z\rangle=+0.056\), so the
factorized term alone, evaluated with the exact \(\langle\sigma_z\rangle\)
and \(Z_1\), gives \(C_2=+0.021\), which \(\operatorname{Im}\delta_2\)
reduces to the exact \(+0.0043\); the closure, however, predicts
\(C=-0.017\) because its self-consistent inversion is still negative,
\(\langle\sigma_z\rangle_{\rm cl}=-0.039\).  Its sign error is thus
primarily indirect --- through the inversion, i.e.\ through \(\delta_1\)
and the photon number --- with the omitted \(\delta_2\) a direct
correction of comparable size.  In the good-cavity regime at \(s=0.9\)
the closure predicts the correct inversion sign and reproduces the
coherence within 9--14\% (\(+0.043\) and \(+0.032\) against the exact
\(+0.048\) and \(+0.037\) for \(N=2,3\)); the remaining discrepancy
combines the omitted \(\delta_2\) contribution with errors in the
self-consistent inversion and atom--field polarization, which partly
cancel (for \(N=2\) the factorized term evaluated on the exact state is
\(+0.058\), \(\operatorname{Im}\delta_2\) brings it to \(+0.048\), and
the closure's own retained term gives \(+0.043\)).  These atom--field cumulants are of the
same family as the purely atomic \(\Gamma_3\) and \(K_3\)
(Fig.~\ref{fig:C3zK3}), and Sec.~\ref{sec:maxent} shows that the
maximum-entropy reconstruction from pair marginals does not recover the
three-body structure they encode.  We note the distinction between
the pump parameter \(s>1/2\), at which the \emph{uncoupled} atoms would
be inverted and which we call strong pump, and actual inversion
\(\langle\sigma_z\rangle>0\) of the coupled system, reached only at
larger \(s\) (Table~\ref{tab:omitted}: \(\langle\sigma_z\rangle<0\) at
\(s=0.6\)).

\begin{table*}[tbp]
\caption{The two cumulants set to zero by the second-order closure,
Eq.~(\ref{eq:factorization}), evaluated on the exact stationary state at
\(\Delta=0\), \(\nu=0\): \(\delta_1=\langle a^\dagger a\,\sigma_z^{(j)}\rangle-\langle\nop\rangle\langle\sigma_z\rangle\)
with the factorized value \(\langle\nop\rangle\langle\sigma_z\rangle\), and
the two terms of Eq.~(\ref{eq:C2exact}),
\(\langle\sigma_z\rangle\operatorname{Im}Z_1\) and
\(\operatorname{Im}\delta_2\) (both real; \(Z_1\) and \(\delta_2\) are
purely imaginary at resonance), whose sum times \(-g/2C_{\rm r}\)
reproduces the exact \(C_2^{(N)}\) in the next column.  The last two
columns are the closure's self-consistent inversion and pair coherence.
Baseline parameters for \(N=3\); good-cavity parameters of
Sec.~\ref{sec:goodcavity} for the last two rows.  The \(N=3\) baseline
values agree with brute-force diagonalization to \(10^{-13}\) or
better.}
\label{tab:omitted}
\begin{ruledtabular}
\begin{tabular}{cccrrrrrrr}
\(N\) & \(s\) & \(A\) & \(\langle\nop\rangle\langle\sigma_z\rangle\) & \(\delta_1\) & \(\langle\sigma_z\rangle\operatorname{Im}Z_1\) & \(\operatorname{Im}\delta_2\) & \(C_2^{(N)}\) exact & \(\langle\sigma_z\rangle_{\rm cl}\) & \(C_{\rm cl}\)\\
\hline
3 & 0.3 & 1 & \(-0.0805\) & \(-0.0147\) & \(+0.0244\) & \(-0.0039\) & \(-0.0322\) & \(-0.543\) & \(-0.0387\)\\
3 & 0.6 & 1 & \(-0.0742\) & \(-0.0855\) & \(+0.0225\) & \(-0.0035\) & \(-0.0298\) & \(-0.217\) & \(-0.0454\)\\
3 & 0.9 & 1 & \(+0.0440\) & \(-0.2150\) & \(-0.0133\) & \(+0.0106\) & \(+0.0043\) & \(-0.039\) & \(-0.0165\)\\
2 & 0.9 & 0.1 & \(+1.1629\) & \(-0.4660\) & \(-0.1163\) & \(+0.0211\) & \(+0.0476\) & \(+0.129\) & \(+0.0434\)\\
3 & 0.9 & 0.1 & \(+1.3166\) & \(-0.4852\) & \(-0.0878\) & \(+0.0135\) & \(+0.0371\) & \(+0.090\) & \(+0.0318\)\\
\end{tabular}
\end{ruledtabular}
\end{table*}

\begin{figure}[tbp]
\includegraphics[width=\linewidth]{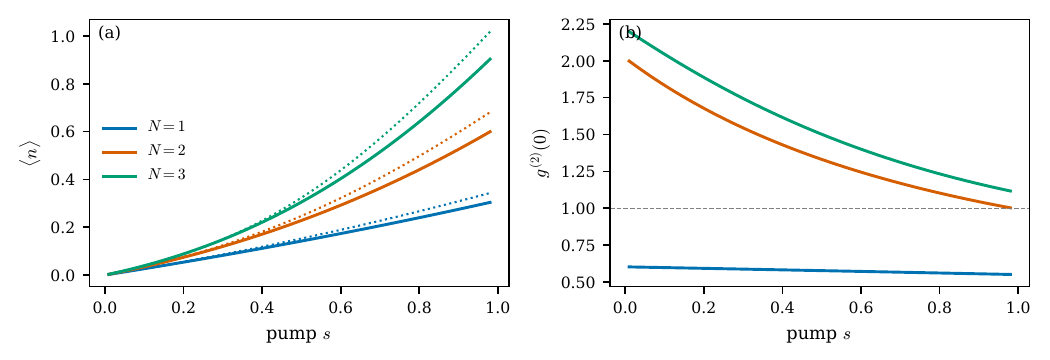}
\caption{(a) Stationary photon number and (b) \(g^{(2)}(0)\) versus pump
for \(N=1,2,3\) at the baseline parameters.  Dotted in (a): second-order
cumulant closure, Eq.~(\ref{eq:closure}).}
\label{fig:n_g2}
\end{figure}

\section{Emission spectra and few-emitter lasing}
\label{sec:lasing}

\subsection{Bad-cavity regime: collective Rabi doublet}
\label{sec:spectrum-baseline}

Figure~\ref{fig:spectrum} shows the dominant pole of \(g^{(1)}(\tau)\) in
the baseline regime.  The spectrum is a doublet: \(\lambda_1\) has an
imaginary part that at weak pump equals the collective vacuum-Rabi
splitting, growing from \(0.54A\) at \(N=1\) to \(1.05A\) at \(N=4\)
(\(g\sqrt N/2\) would give \(0.55,0.78,0.95,1.10\)), and that collapses
toward zero as the pump saturates the atoms.  The width of each doublet
component at weak pump is \(A/2+C=0.85A\), the sum of the cavity and
polarization amplitude-decay rates, and grows with \(N\) and pump to
\(1.19A\) before dropping where the doublet merges.  The spectral
function itself (Fig.~\ref{fig:spectra}, left and center) shows the
doublet at \(s=0.3\) --- clearly double-peaked for \(N=2,3\),
flat-topped for \(N=1\) --- of total FWHM \(1.3\)--\(2.1A\) for
\(N=1\)--3, and the merged single peak at \(s=0.9\), whose FWHM
(\(0.74\)--\(0.83A\)) is \emph{smaller} than the width of either
overlapping pole (\(0.94\)--\(1.12A\)) because the two components have
complex residues.  The reduction of the spectral span from
\(s=0.3\) to \(s=0.9\) accompanies the collapse of the doublet; this
regime does not exhibit the sustained, lasing-associated narrowing of a
single line observed in the good-cavity calculations below --- the
baseline system of Sec.~\ref{sec:collective} is a strongly coupled
few-quanta regime, \(\langle\nop\rangle<1\), \(g^{(2)}(0)>1\).

\begin{figure}[tbp]
\includegraphics[width=\linewidth]{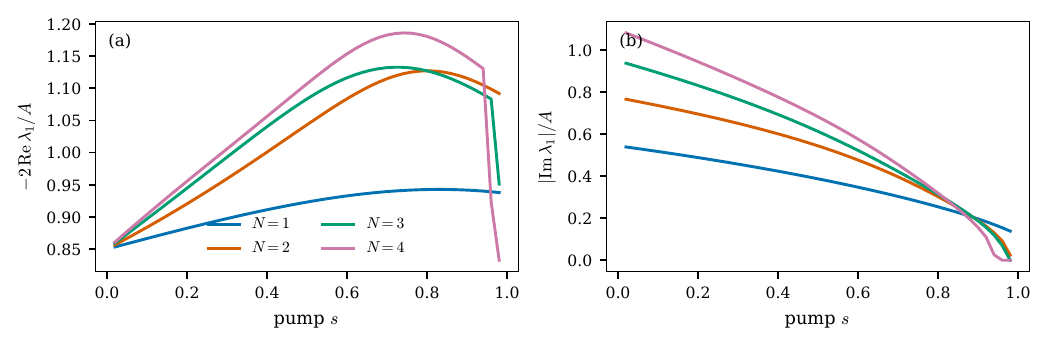}
\caption{Dominant pole of \(g^{(1)}(\tau)\) at the baseline parameters for
\(N=1,\ldots,4\): (a) width \(-2\operatorname{Re}\lambda_1\) of one doublet
component, (b) offset \(|\operatorname{Im}\lambda_1|\) from the cavity
frequency.}
\label{fig:spectrum}
\end{figure}

\subsection{Good-cavity regime: lasing crossover and line narrowing}
\label{sec:goodcavity}

To reach a lasing regime we take \(A=0.1\), \(B=1\), \(C=B/2\), \(g=0.5\)
(all rates in the same units), so that the cavity decays slowly compared
with the atomic polarization and \(g^2/(AC)=5\).  We use ``few-emitter
lasing'' operationally for the regime in which increasing pump produces a
rapid rise of the cavity population, near-Poissonian photon statistics,
and substantial spectral narrowing; at finite \(N\) there is no sharp
threshold, only a threshold-like crossover.  Figure~\ref{fig:goodcavity}
shows, for \(N=1,\ldots,5\), such a crossover at strong pump, \(s>1/2\) --- to
\(3.1,7.7,12.4,17.2,22.0\) photons at \(s=0.98\), i.e.\ 4.4 per atom at \(N=5\) --- a
transition of \(g^{(2)}(0)\) from bunching (\(1.6\)--\(2.0\)) to
\(1.008,1.005,1.003\) (\(N=1,2,3\)), and a positive pair coherence above
the crossover (\(C_2^{(2)}=+0.064\), \(C_2^{(5)}=+0.032\)).  The spectrum
evolves from a Rabi doublet of component width \(\approx A/2+C=5.5A\) at
weak pump, through a doublet collapse near \(s\approx0.2\), to a single
line that narrows monotonically: at \(s=0.9\) its width is
\(0.411A,0.249A,0.198A\) for \(N=1,2,3\), and at \(s=0.98\)
\(0.31A,0.19A,0.16A\); in this single-line regime, for \(s\ge0.6\), the
width of the spectral function agrees with the pole width to 0.1--0.5\%
(Table~\ref{tab:spectral-checks}; Fig.~\ref{fig:spectra}, right).  The scale \(A/(2\langle\nop\rangle)\)
gives \(0.21A,0.08A,0.05A\) at \(s=0.9\); the exact line is 2, 3, and 4
times broader, a discrepancy that quantifies the importance of
finite-gain-medium and pump-induced noise beyond the simple
\(1/\langle\nop\rangle\) estimate.  Table~\ref{tab:goodcavity} collects
these numbers.

\begin{table*}[tbp]
\caption{Good-cavity regime (\(A=0.1\), \(B=1\), \(C=B/2\), \(g=0.5\)) at
\(s=0.9\): exact photon number, \(g^{(2)}(0)\), pair coherence, second-order
closure, and the width of the dominant single-line pole compared with
\(A/(2\langle\nop\rangle)\).  Linewidth entries are omitted where the pole
decomposition is beyond the conditioning limit of
Sec.~\ref{sec:precision}.}
\label{tab:goodcavity}
\begin{ruledtabular}
\begin{tabular}{ccccccc}
\(N\) & \(\langle\nop\rangle\) & \(g^{(2)}(0)\) & \(C_2^{(N)}\) & \(\langle\nop\rangle_{\rm closure}\) & \(-2\operatorname{Re}\lambda_1/A\) & \(1/(2\langle\nop\rangle)\)\\
\hline
1 & 2.4217 & 1.0629 & --- & 2.869 & 0.411 & 0.206\\
2 & 6.0908 & 1.0376 & \(+0.0476\) & 6.707 & 0.249 & 0.082\\
3 & 10.031 & 1.0212 & \(+0.0371\) & 10.656 & 0.198 & 0.050\\
4 & 14.012 & 1.0143 & \(+0.0300\) & 14.633 & --- & 0.036\\
5 & 18.002 & 1.0108 & \(+0.0250\) & 18.619 & --- & 0.028\\
\end{tabular}
\end{ruledtabular}
\end{table*}

\begin{figure*}[tbp]
\includegraphics[width=\textwidth]{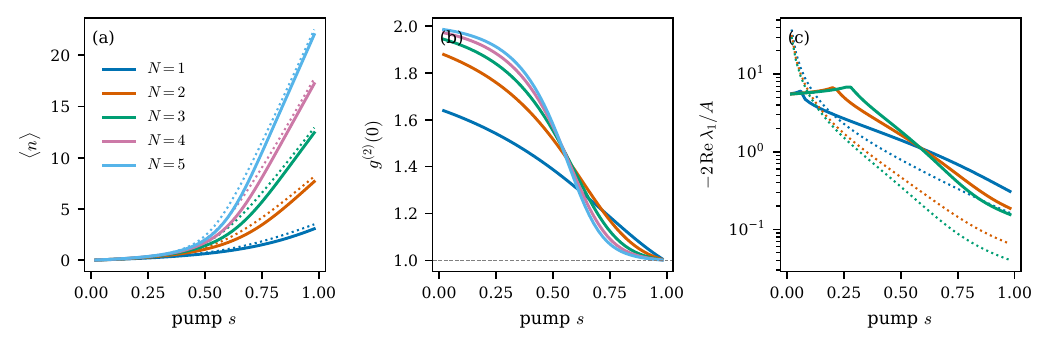}
\caption{Good-cavity regime, \(A=0.1\), \(B=1\), \(C=B/2\), \(g=0.5\),
\(\Delta=0\), \(\nu=0\), for \(N=1,\ldots,5\): (a) photon number (solid:
exact; dotted: second-order closure), (b) \(g^{(2)}(0)\), (c) width of the
dominant pole of \(g^{(1)}(\tau)\) for \(N=1,2,3\) (solid) and the scale
\(A/(2\langle\nop\rangle)\) (dotted).  The kink near \(s\approx0.2\) in (c)
is the collapse of the weak-pump Rabi doublet into a single line.}
\label{fig:goodcavity}
\end{figure*}

\begin{figure*}[tbp]
\includegraphics[width=\textwidth]{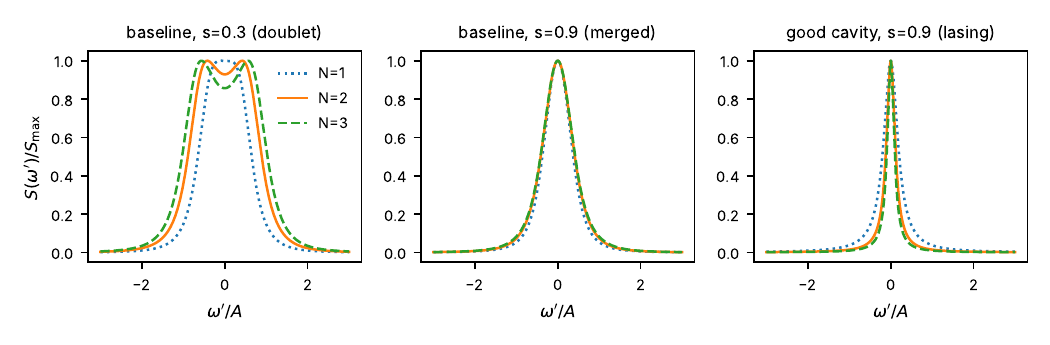}
\caption{Normalized emission spectra \(S(\omega')/S_{\max}\) for
\(N=1,2,3\), \(\Delta=0\), \(\nu=0\), with \(\omega'\) the detuning from
the cavity frequency: collective Rabi doublet at the baseline
parameters and \(s=0.3\) (FWHM \(1.30A,1.78A,2.06A\) for \(N=1,2,3\),
versus one-component pole widths \(0.90A,0.96A,0.99A\); the poles are
split for all three \(N\), but for \(N=1\) the two components overlap
so strongly that the spectrum is flat-topped rather than visibly
double-peaked --- split poles do not by themselves imply resolved
peaks);
merged line at the baseline parameters and \(s=0.9\) (FWHM
\(0.74A,0.82A,0.83A\), narrower than the pole widths
\(0.94A,1.12A,1.10A\)); single narrowed line in the good-cavity regime
of Sec.~\ref{sec:goodcavity} at \(s=0.9\) (FWHM \(0.410A,0.250A,0.198A\),
equal to the pole widths \(0.411A,0.249A,0.198A\)).}
\label{fig:spectra}
\end{figure*}

This is the setting of Yeoman and Meyer's two-atom laser
\cite{YeomanMeyer1998}, who compared one and two atoms at equal per-atom
rates and attributed the improved two-atom line to the cavity-induced
mutual coherence.  The exact calculation confirms the improvement: adding
the second atom narrows the line by a factor 1.65 and raises the photon
number by 2.5.  Equation~(\ref{eq:coherenceflux}) is a photon-number
balance, not a linewidth identity, so it cannot by itself assign the
narrowing to a mechanism; what it shows is that the regime in which the
line improves is also one in which the pair coherence contributes
positively to cavity feeding --- \(+8\%\) of the population term at
\(s=0.9\), \(+11\%\) at \(s=0.98\), against \(-28\%\) in the bad-cavity
regime of Table~\ref{tab:balance} --- which supports the interpretation
proposed by Yeoman and Meyer.  The narrowing is weaker than the
\(1/\langle\nop\rangle\) scaling would predict (factor 2.5), which is the
quantitative content of ``few-atom'' in few-atom laser.

\section{Discussion and scope}
\label{sec:discussion}

\paragraph{What is new relative to permutation-invariant Fock-space solvers.}
Both use the same \(D_N\)-dimensional atomic space
\cite{ChaseGeremia2008,XuTieriHolland2013,GeggRichter2016,Shammah2018}.
The Fock-space solvers represent the field by a cutoff \(n_{\rm cav}\),
giving a sparse Liouvillian of dimension \(D_Nn_{\rm cav}^2\) before,
and --- since the same \(U(1)\) symmetry restricts the stationary sector
to Fock elements \(|n_{\rm L}\rangle\langle n_{\rm R}|\) with
\(n_{\rm L}-n_{\rm R}=-q_\alpha\) --- of dimension
\(\sum_\alpha\max\{n_{\rm cav}-|q_\alpha|,0\}\approx D_Nn_{\rm cav}\)
after exploiting the conserved coherence, whose nullspace is found by
sparse linear algebra.  The reduction to a single field index is
therefore not an advantage specific to the damping basis: a permutation-
and \(U(1)\)-reduced Fock solver works in a space of comparable size,
and the brute-force comparisons of this paper (dimension
\(4^Nn_{\rm cav}^2\) without either reduction) are checks, not a
benchmark of cost.  What the damping basis adds is the structure: a
block-tridiagonal operator with \(D_N\times D_N\) blocks, a continued
fraction with a provably unique linear closure that converges to the
stationary state at \(\nu=0\) under hypotheses (U) and (P)
(Theorem~\ref{thm:branch}), the grading behind it, moment functionals
that read off observables from the lowest radial coefficients
(Sec.~\ref{sec:observables}), and a weak-pump series that terminates
exactly in the radial index at each order (Sec.~\ref{sec:hierarchy}).  A
timing benchmark against a \(U(1)\)-reduced Fock solver has not been
performed.  The Fock solvers have no
\(e^{\langle\nop\rangle}\) conditioning limit and are preferable for
\(\langle\nop\rangle\gtrsim20\).

\paragraph{What is new relative to one-atom damping-basis work.}
Refs.~\cite{BriegelEnglert1993,Ginzel1993} solved the one-atom laser
with a scalar-block recurrence.  The point of the present construction is
that permutation symmetry keeps the many-atom sector finite while
preserving the nearest-neighbor radial structure, so that the one-atom
continued fraction generalizes with \(D_N\times D_N\) blocks and no new
algebra; the grading of Proposition~\ref{prop:grading}, which makes the
closure linear and the asymptotic transfer matrix nilpotent, is present
already at \(N=1\) (index 2) but its consequences for branch selection
had not been drawn.

\paragraph{Why it matters physically.}
The construction provides an exact benchmark in precisely the finite-\(N\)
regime where cumulant truncations are least controlled and where
few-emitter lasers operate.  Section~\ref{sec:benchmark} shows that the
second-order closure is semi-quantitative for the photon number but
misses the sign of the inter-emitter coherence at strong pump, and
Secs.~\ref{sec:collective} and \ref{sec:lasing} show what the closure
cannot see: the singlet-controlled sign of the two-atom coherence, the
sector parity, the pump-order hierarchy, and the full spectrum.

\paragraph{Fast-atomic-correlation limit.}
When atomic correlations decay much faster than the field and atom--field
factorization is assumed, \(\langle S_+S_-\rangle\to Ns\) and
\(\langle\nop S_z\rangle\to N(2s-1)\langle\nop\rangle\), and
Eq.~(\ref{eq:flux2}) reduces to
\(A(\langle\nop\rangle-\nu)=N\mathcal R[s+(2s-1)\langle\nop\rangle]\), the
photon-number balance of the effective Liouvillian
\(\mathcal L_{\rm f}+N\mathcal Rs\,\mathcal D[a^\dagger]+N\mathcal R(1-s)\mathcal D[a]\),
with \(\mathcal R\to g^2C/[2(C^2+\Delta^2)]\) for \(C\gg A\).  This is
Fleischhauer's \(N\)-times-one-atom scaling \cite{Fleischhauer1994} at
the level of the first moment (for \(A=1\), \(B=20\), \(C=10\), \(s=0.2\),
\(g=0.5\) the exact photon numbers differ from it by
\(8.2\times10^{-4}\), \(1.17\times10^{-3}\), \(1.52\times10^{-3}\) for
\(N=1,2,3\)); it is a consistency check, not a derivation of the complete
effective superoperator.

\paragraph{Limitations.}
The construction requires identical emitters with identical couplings and
symmetric local rates: no inhomogeneous broadening, no emitter-dependent
coupling, and no collective decay channel other than the one mediated by
the cavity (a collective \(\mathcal D[S_-]\) term would preserve the
symmetry and could be added).  Its cost is polynomial but steep,
\(O(n_{\max}D_N^3)\sim O(N^9)\) with dense blocks, and its double-precision
accuracy degrades empirically as \(\epsilon e^{\langle\nop\rangle}\) with the
photon number and, through the conditioning of the pivots, with \(N\)
(at \(N=16\) the plain recursion exhibits relative variations up to
a few \(10^{-3}\), while the grade rescaling of Sec.~\ref{sec:precision}
reduces the cutoff sensitivity to about \(10^{-8}\); the latter is a
reproducibility result, and the absolute accuracy at \(N=16\) has not
been independently established).  The convergence
theorem for the continued fraction is proved for a zero-temperature
cavity and rests on hypothesis (U) and on the invertibility of the
finitely many non-asymptotic pivots (Remark~\ref{rem:status}); a thermal
cavity and the regression spectrum are validated numerically only.  Open
problems are a proof of (U) from the semigroup literature for this
model, the extension of Theorem~\ref{thm:branch} to \(\nu>0\), a
criterion for the nonvanishing of the leading weak-pump coefficients
of Table~\ref{tab:series} for general \(m\) and \(N\), the weak-pump
asymptotics of \(I_3^{\rm irr}\), a graded norm in which the
conditioning of the pivots would be controlled, and a displaced or
squeezed damping basis that would lift the \(e^{\langle\nop\rangle}\)
limit.

\section{Conclusion}
\label{sec:conclusion}

Combining the field damping basis with the permutation-invariant atomic
basis turns the pumped finite-\(N\) Tavis--Cummings problem into an exact
block-Jacobi problem whose stationary state is a matrix continued
fraction with a linear closure that is uniquely solvable for \(A>0\);
for a zero-temperature cavity the backward recursion converges
independently of its terminal condition, and under a moment-stationarity
uniqueness hypothesis (U) and an invertibility condition (P) on the
limiting pivots its limit is the coefficient sequence of the unique
stationary state of the infinite-dimensional problem, so that the
remaining approximations are the finite radial truncation and
floating-point arithmetic, whose effects we quantify.  The emission
spectrum follows from the same recurrence.  The exact solution
gives the two-atom photon balance with the pair coherence as an explicit
term tied to the singlet population, proved pump-order upper bounds for
excitation cumulants and connected coherences, attained with nonzero
leading coefficients through eighth order at the parameters studied, a
parity theorem for the
minimal collective-spin
sector, and a quantitative benchmark showing that the second-order
cumulant closure, while adequate for the photon number, misses the sign
of the inter-emitter coherence at strong pump.  In a good-cavity regime
the same model shows near-Poissonian statistics and a line narrowing to
\(0.2A\) for three emitters.  The method's conditioning limit places it
in the few-quanta regime where microlasers operate and where exact
benchmarks are most needed.

\section*{Data availability}
The code and numerical data necessary to reproduce the results, including
figure-generation scripts and regression tests, are openly available in
the Zenodo repository of Ref.~\cite{code}
(\url{https://doi.org/10.5281/zenodo.23313491}, version 2.1).

\begin{acknowledgments}
The author thanks Hans Briegel for valuable feedback on an earlier version of this manuscript, which was greatly inspired by his pioneering work with Englert in 1993.

The author also thanks the anonymous referee for a careful and very helpful report, which led to substantial improvements of the paper.

Claude Fable 5 (Anthropic) was used as an assistive tool for literature searches, manuscript editing, code generation, and symbolic and numerical checks, including the preparation of plots. All AI-assisted code and calculations were independently checked, and all references and text were independently reviewed by the author, who takes full responsibility for the manuscript.
\end{acknowledgments}

\appendix

\section{Field multiplication coefficients}
\label{app:fieldcoeff}

The Briegel--Englert coefficients of Eq.~(\ref{eq:field-action-coeff}),
with \(\rho_{-1}^{(k)}=0\).  For \(k>0\),
\begin{align}
a\rho_n^{(k)}&=(n+k)\rho_n^{(k-1)}+\frac{(n+1)\nu}{1+\nu}\rho_{n+1}^{(k-1)},\label{eq:fa1}\\
a^\dagger\rho_n^{(k)}&=(1+\nu)\rho_{n-1}^{(k+1)}+(1+\nu)\rho_n^{(k+1)},\label{eq:fa2}\\
\rho_n^{(k)}a&=(n+k)\rho_n^{(k-1)}+(n+1)\rho_{n+1}^{(k-1)},\label{eq:fa3}\\
\rho_n^{(k)}a^\dagger&=(1+\nu)\rho_{n-1}^{(k+1)}+\nu\rho_n^{(k+1)};\label{eq:fa4}
\end{align}
for \(k<0\),
\begin{align}
a\rho_n^{(k)}&=(1+\nu)\rho_{n-1}^{(k-1)}+\nu\rho_n^{(k-1)},\label{eq:fb1}\\
a^\dagger\rho_n^{(k)}&=(n+|k|)\rho_n^{(k+1)}+(n+1)\rho_{n+1}^{(k+1)},\label{eq:fb2}\\
\rho_n^{(k)}a&=(1+\nu)\rho_{n-1}^{(k-1)}+(1+\nu)\rho_n^{(k-1)},\label{eq:fb3}\\
\rho_n^{(k)}a^\dagger&=(n+|k|)\rho_n^{(k+1)}+\frac{(n+1)\nu}{1+\nu}\rho_{n+1}^{(k+1)};\label{eq:fb4}
\end{align}
and at \(k=0\),
\begin{align}
a\rho_n^{(0)}&=(1+\nu)\rho_{n-1}^{(-1)}+\nu\rho_n^{(-1)},\label{eq:fc1}\\
a^\dagger\rho_n^{(0)}&=(1+\nu)\rho_{n-1}^{(1)}+(1+\nu)\rho_n^{(1)},\label{eq:fc2}\\
\rho_n^{(0)}a&=(1+\nu)\rho_{n-1}^{(-1)}+(1+\nu)\rho_n^{(-1)},\label{eq:fc3}\\
\rho_n^{(0)}a^\dagger&=(1+\nu)\rho_{n-1}^{(1)}+\nu\rho_n^{(1)}.\label{eq:fc4}
\end{align}
The normalization is fixed by these identities together with
\(\rho_0^{(0)}=\) the thermal state of mean \(\nu\); they imply
\(\rho_n^{(k+1)}=[a^\dagger,\rho_n^{(k)}]\) for \(k\ge0\),
\(\rho_n^{(k-1)}=-[a,\rho_n^{(k)}]\) for \(k\le0\), and
\(\rho_{n+1}^{(0)}=-\tfrac{1+\nu}{n+1}[a^\dagger,[a,\rho_n^{(0)}]]\).
Constructing the \(\rho_n^{(k)}\) from these generating relations in a
60-level Fock space, we have verified that all twelve identities, the
eigenvalue equation (\ref{eq:fieldeigen}), and
\(\operatorname{tr}\rho_n^{(k)}=\delta_{n0}\delta_{k0}\) hold to \(10^{-14}\)
at \(\nu=0\) and \(\nu=0.3\).

\section{Observable functionals}
\label{app:observables}

Let \(c_n\) be the component of \(X_n\) along \((N,0,0,0)\).  Only this
component has a nonvanishing atomic trace, so it alone determines the
\(k=0\) expansion of the reduced field state, and
Eqs.~(\ref{eq:fc1})--(\ref{eq:fc4}) give
\(\operatorname{tr}(\nop\rho_n^{(0)})=\nu\delta_{n0}+(1+\nu)\delta_{n1}\) and
\(\operatorname{tr}(a^{\dagger2}a^2\rho_n^{(0)})=2\nu^2\delta_{n0}+4\nu(1+\nu)\delta_{n1}+2(1+\nu)^2\delta_{n2}\),
hence
\begin{align}
\langle\nop\rangle&=\nu+(1+\nu)c_1,\notag\\
\langle\nop(\nop-1)\rangle&=2\nu^2+4\nu(1+\nu)c_1+2(1+\nu)^2c_2,
\label{eq:obs-n}
\end{align}
and more generally \(\langle a^{\dagger m}a^m\rangle=m!\,c_m\) at \(\nu=0\).
For the pair coherence, only \(\mv=(N-2,0,1,1)\) contributes to
\(\langle\tau_+^{(1)}\tau_-^{(2)}\rangle\), and of its \(N(N-1)\)
arrangements exactly one places \(\tau_-\) on atom 1 and \(\tau_+\) on
atom 2, so \(C_2^{(N)}=X_0[(N-2,0,1,1)]/[N(N-1)]\).  For the joint
functional Eq.~(\ref{eq:joint-observable}), Eq.~(\ref{eq:fb4}) gives
\(\rho_n^{(-1)}a^\dagger=(n+1)\rho_n^{(0)}+\tfrac{(n+1)\nu}{1+\nu}\rho_{n+1}^{(0)}\),
whose trace is \(\delta_{n0}\).  For the spectrum
(Sec.~\ref{sec:spectrum}) the initial operator \(a\rho_{\rm ss}\) has
\(K=-1\) coefficients \(Y\) given by
\((Y_{n+\delta})_\alpha\mathrel{+}=f^{\rm L}_{-,\delta}(n,-q_\alpha)(X_n)_\alpha\),
and \(\operatorname{tr}(a^\dagger\Phi^{(-1)}_{n\alpha})=\delta_{n0}\delta_{\alpha,(N,0,0,0)}\).

\section{Proof of Theorem~\ref{thm:branch}}
\label{app:proof}

Throughout, \(\nu=0\), \(K=0\), and (U), (P) are assumed.  Let
\(X^{\rm ph}_{n\alpha}=\Tr(\rho_{\rm ss}\check\Phi_{n\alpha})\) and call a
solution \(Y\) of Eq.~(\ref{eq:stationary}) with \(Y_{-1}=0\)
\emph{super-exponentially decaying} if \(\|Y_n\|r^n\to0\) for every
\(r>0\).

\emph{Step 1: \(X^{\rm ph}\) is a super-exponentially decaying solution.}
Its components exist and decay super-exponentially by
Proposition~\ref{prop:fock} and the Cauchy--Schwarz bound following
Lemma~\ref{lem:expmoments}.  For each dual element,
\(\Tr(\rho^{(n_{\rm c})}\mathcal L^\dagger_{n_{\rm c}}\check\Phi_{m\beta})=0\),
and \((\mathcal L^\dagger_{n_{\rm c}}-\mathcal L^\dagger)\check\Phi_{m\beta}\)
is supported on Fock levels above \(n_{\rm c}-m-N-1\) with matrix
elements bounded by a polynomial in \(n_{\rm c}\); by the uniform
super-exponential tails its expectation vanishes as \(n_{\rm c}\to\infty\),
while \(\Tr(\rho^{(n_{\rm c})}\mathcal L^\dagger\check\Phi_{m\beta})\to
\Tr(\rho_{\rm ss}\mathcal L^\dagger\check\Phi_{m\beta})\) by
Proposition~\ref{prop:fock}.  Hence \(\rho_{\rm ss}\) satisfies the moment
equations, which are the recurrence.

\emph{Step 2: every super-exponentially decaying solution is a multiple
of \(X^{\rm ph}\).}  Let \(Y\) be one.  At \(\nu=0\),
\(\rho_n^{(0)}=\sum_{m\le n}(-1)^{n-m}\binom nm|m\rangle\langle m|\) and
\(\rho_n^{(k)}\) is obtained from it by \(|k|\) commutators with
\(a^\dagger\) or \(a\), so \(\rho_n^{(k)}\) has finite rank, is supported
on Fock levels \(\le n+|k|\), and
\(\|e^{\theta\nop}\rho_n^{(k)}e^{\theta\nop}\|_1\le
e^{2\theta(n+N)}(2\sqrt{n+N+1})^{N}2^{n}\).  Therefore
\(\sigma=\sum_{n,\alpha}Y_{n\alpha}\rho_n^{(k_\alpha)}\otimes R_{\mv_\alpha}\)
converges absolutely in trace norm together with
\(e^{\theta\Exc}\sigma e^{\theta\Exc}\) for every \(\theta\), i.e.\
\(\sigma\in\mathcal T_{\exp}\).  Pairing with a fixed dual element is
then legitimate termwise: \(\check\Phi_{m\beta}\) is a polynomial of
fixed degree, so \(\|e^{-\theta\Exc}\check\Phi_{m\beta}e^{-\theta\Exc}\|<\infty\)
for any \(\theta>0\) and
\(|\Tr(\rho_n^{(k)}\otimes R_{\mv}\,\check\Phi_{m\beta})|\le
\|e^{-\theta\Exc}\check\Phi_{m\beta}e^{-\theta\Exc}\|\,\|e^{\theta\Exc}\rho_n^{(k)}\otimes R_{\mv}e^{\theta\Exc}\|_1\),
whose sum over \(n,\alpha\) weighted by \(|Y_{n\alpha}|\) converges by
the bound above; dominated convergence and biorthogonality give
\(\Tr(\sigma\check\Phi_{n\alpha})=Y_{n\alpha}\).  Thus the
moment sequence of \(\sigma\) is \(Y\), the recurrence for \(Y\) is the
statement \(\Tr(\sigma\,\mathcal L^\dagger\check\Phi_{m\beta})=0\) for
all \(m,\beta\).  The \(\check\Phi_{m\beta}\) of the sector \(K=0\) span
the permutation-invariant polynomials of total \(U(1)\) charge zero; for
a polynomial \(Y'\) of nonzero charge \(K'\), \(\mathcal L^\dagger Y'\)
has charge \(K'\) by the \(U(1)\) covariance of Eq.~(\ref{eq:master}) and
\(\Tr(\sigma\,\mathcal L^\dagger Y')=0\) because \(\sigma\) has charge
zero (it is invariant under \(e^{i\varphi\Exc}\cdot e^{-i\varphi\Exc}\)),
and for a non-invariant \(Y'\) the statement follows by averaging over
permutations, under which \(\sigma\) and \(\mathcal L\) are invariant.
Hence \(\Tr(\sigma\,\mathcal L^\dagger Y)=0\) for all polynomial \(Y\),
and hypothesis (U) gives
\(\sigma=\lambda\rho_{\rm ss}\) and \(Y=\lambda X^{\rm ph}\).

\emph{Step 3: the physical solution follows the limiting transfer
matrices.}  Fix \(n\ge0\) and \(n_{\max}>n\), and let
\(P_j=M_j+GR_j^{(n_{\max})}\) be the pivots of the recursion from
\(R_{n_{\max}}=R_\infty\).  The equations of Eq.~(\ref{eq:stationary}) at
levels \(n+1,\ldots,n_{\max}\) are satisfied by \(X^{\rm ph}\) with the
terminal relation
\(X^{\rm ph}_{n_{\max}+1}=R_\infty X^{\rm ph}_{n_{\max}}+Z\),
\(Z=X^{\rm ph}_{n_{\max}+1}-R_\infty X^{\rm ph}_{n_{\max}}\).  Back
substitution gives \(X^{\rm ph}_{j}=R^{(n_{\max})}_{j-1}X^{\rm ph}_{j-1}+T_jZ\)
with \(T_{n_{\max}+1}=\mathbb I\), \(T_j=-P_j^{-1}GT_{j+1}\), so that
\(X^{\rm ph}_{n+1}=R_n^{(n_{\max})}X^{\rm ph}_n+T_{n+1}Z\).  By the
estimates in the proof of Proposition~\ref{prop:terminal},
\(\|P_j^{-1}\|\|G\|\le C''/j\) for \(j>n_0\), with
\(C''=2\|(M^{(1)})^{-1}\|\|G\|\).  For the finitely many levels
\(j\le n_0\), hypothesis (P) says that the limits \(M_j+GR_j^\infty\) are
invertible, and since \(R_j^{(n_{\max})}\to R_j^\infty\)
(Proposition~\ref{prop:terminal}) the inverses \(P_j^{-1}\) are bounded
uniformly in \(n_{\max}\) for all sufficiently large \(n_{\max}\), by
continuity of the inverse.  For the fixed level \(n\) put
\(m=\max\{n,n_0\}\) and let \(C_n\) be the product of the bounds on
\(\|P_j^{-1}\|\|G\|\) over the finitely many levels \(n<j\le m\) (empty
product one if \(n\ge n_0\)).  Then, with \(L=n_{\max}\),
\(\|T_{n+1}\|\le C_n\,(C'')^{\,L-m}\,m!/L!\), while \(\|Z\|\) stays
bounded (indeed tends to zero).  Letting
\(n_{\max}\to\infty\) yields \(X^{\rm ph}_{n+1}=R_n^\infty X^{\rm ph}_n\)
for all \(n\ge0\); the level-0 equation then reads
\((M_0+GR_0^\infty)X^{\rm ph}_0=0\), and \(X^{\rm ph}_0\neq0\) because its
trace component is one.

\emph{Step 4: the kernel is one-dimensional and the recursion
converges.}  Let \(v\in\ker(M_0+GR_0^\infty)\) and define
\(Y_0=v\), \(Y_{n+1}=R_n^\infty Y_n\).  The level-\(j\) equation,
\(j\ge1\), reads \([(M_j+GR_j^\infty)R_{j-1}^\infty+F_j]Y_{j-1}=0\), which
holds because \(R_{j-1}^\infty=\Phi_{j-1}(R_j^\infty)\) in the limit (the
maps are continuous where the pivots are invertible, by (P) and the
asymptotic estimates).  Since \(R_n^\infty=R_\infty+O(1/n)\) and
\(R_\infty^{\,N+1}=0\), expanding the product of \(d=N+1\) consecutive
factors shows that every term containing only \(R_\infty\) vanishes
and every other term carries at least one \(O(1/n)\) factor, so
\(\|R^\infty_{n+d-1}\cdots R^\infty_n\|\le C_3/n\) for \(n\ge n_1\),
where we choose \(n_1\ge d\); iterating and using
\(\prod_{l=0}^{k-1}(n_1+ld)\ge d^kk!\),
\(\|Y_{n_1+kd}\|\le\|Y_{n_1}\|\prod_{l=0}^{k-1}C_3/(n_1+ld)\le
\|Y_{n_1}\|\,C_3^{\,k}/(d^{k}k!)\), a factorial bound; the intermediate
indices \(n_1+kd<n<n_1+(k+1)d\) are controlled by the uniformly
bounded products of the fewer than \(d\) remaining transfer matrices.
Hence \(Y\) decays faster than any geometric sequence, and by Step~2,
\(v\in\mathbb C\,X^{\rm ph}_0\).
Hence \(\ker(M_0+GR_0^\infty)\) is spanned by \(X^{\rm ph}_0\), the
second-smallest singular value of \(M_0+GR_0^\infty\) is positive, and
the normalized right singular vector of \(M_0+GR_0^{(n_{\max})}\) for its
smallest singular value --- which is how \(X_0^{(n_{\max})}\) is defined
and computed, Sec.~\ref{sec:closure} --- converges to \(X^{\rm ph}_0\) as
\(R_0^{(n_{\max})}\to R_0^\infty\) (Proposition~\ref{prop:terminal}),
by continuity of singular vectors belonging to simple singular values
(up to a phase, which the trace normalization fixes).
Finally \(X_n^{(n_{\max})}=R_{n-1}^{(n_{\max})}\cdots R_0^{(n_{\max})}X_0^{(n_{\max})}\to
R_{n-1}^\infty\cdots R_0^\infty X^{\rm ph}_0=X^{\rm ph}_n\) for each
\(n\), and every observable of Sec.~\ref{sec:observables} is a finite
combination or ratio of finitely many components.  \(\square\)

\section{Validation}
\label{app:validation}

All checks are part of the test suite of Ref.~\cite{code}, which also
verifies Eqs.~(\ref{eq:SLplus})--(\ref{eq:SRminus}) coefficient by
coefficient against explicit tensor-product operators for \(N\le5\),
Eq.~(\ref{eq:M1-structure}) and the nilpotency of \(\mathcal N_K\) for
\(N\le6\), the grading of Proposition~\ref{prop:grading},
Eqs.~(\ref{eq:rhoN-reconstruction}) and
(\ref{eq:joint-observable}) against brute-force partial traces, and the
leading coefficients of Table~\ref{tab:series} for \(N=3,4\).

\subsection{Closed-system checks}

With dissipation off: (i) in the \(J=3/2\) sector of the \(N=3\)
Tavis--Cummings Hamiltonian \cite{TavisCummings1968,JaynesCummings1963},
the coupling matrix elements within excitation manifold \(m\) are
\(g_{\rm BE}\sqrt{3m}\), \(2g_{\rm BE}\sqrt{m-1}\), \(g_{\rm BE}\sqrt{3(m-2)}\)
with \(g_{\rm BE}=g/2\); the Hamiltonian built from the brute-force
operators reproduces them for \(m\le4\) to \(2\times10^{-16}\).  (ii) The
dissipation-free \(K=0\) block-Jacobi matrix assembled from \(M_n,G_n,F_n\)
alone is diagonalized and its eigenvalues compared with the Bohr
frequencies of the Tavis--Cummings spectrum.  Since \(R_{\mv}\) is
\(S_N\)-invariant, only frequencies between states in the same \(J\)
sector appear (Schur's lemma); for \(N=2\) the singlet--triplet frequency
is correctly absent.  For \(N=1\) the Jaynes--Cummings spectrum
\(g_{\rm BE}\sqrt m\) is reproduced exactly; for \(N=2,3\) all within-\(J\)
frequencies of manifolds \(m\le4\) (9 and 17) are reproduced to
\(5\times10^{-11}\), independent of \(s\) and of \(n_{\max}\ge8\).

\subsection{Weak-pump checks}

For \(N=3\) at the baseline parameters, second-order perturbation theory
in \(s\) on the full Liouvillian in a truncated Fock space,
\(\mathcal L_0\rho^{(1)}=-\mathcal L_1\rho^{(0)}\),
\(\mathcal L_0\rho^{(2)}=-\mathcal L_1\rho^{(1)}\) under the trace
constraint, gives
\(\lim\langle\nop\rangle/s=\operatorname{tr}(\nop\rho^{(1)})=0.34519766336\)
and \(\lim g^{(2)}(0)=\operatorname{tr}(a^{\dagger2}a^2\rho^{(2)})/[\operatorname{tr}(\nop\rho^{(1)})]^2=2.21448946864\),
independently of Secs.~\ref{sec:model}--\ref{sec:cf}.  The continued
fraction at \(s=3\times10^{-5}\) gives \(0.34521131203\) and \(2.2144369\),
relative deviations \(4.0\times10^{-5}\) and \(2.4\times10^{-5}\), the
expected \(O(s)\) finite-pump corrections; brute-force diagonalization at
the same \(s\) gives \(0.34521131201\) and \(2.21443\), ten and five
digits of agreement respectively (\(g^{(2)}(0)\) being a ratio of two
\(O(s^2)\) quantities).

\subsection{Closure grid and branch uniqueness}

\begin{table*}[tbp]
\caption{Closure grid for \(N=3\), \(A=1\), \(B=0.7\), \(C=B/2\).  The
photon number is unchanged to all digits when \(M^{(1)},F^{(1)}\) are
extracted at radial index 6, 10, 16, or 24, and agrees with brute-force
diagonalization (cavity cutoff 14, converged against 22) to the relative
accuracy in the last column.  \(\sigma_{\min}\) is the smallest singular
value of \(M^{(1)}\).}
\label{tab:closure}
\begin{ruledtabular}
\begin{tabular}{lccccrc}
Regime & \(g\) & \(\Delta\) & \(s\) & \(\langle\nop\rangle\) & \(\sigma_{\min}\) & rel.\ dev.\\
\hline
Weak coupling & 0.3 & 0 & 0.05 & 0.0058440062 & 0.55 & \(2\times10^{-13}\)\\
Strong coupling & 2.4 & 0 & 0.05 & 0.0216696927 & 0.011 & \(2\times10^{-16}\)\\
Baseline, weak pump & 1.1 & 0 & 0.05 & 0.0184163432 & 0.087 & \(1\times10^{-13}\)\\
Baseline, strong pump & 1.1 & 0 & 0.90 & 0.7808724757 & 0.088 & \(1\times10^{-15}\)\\
Off resonance & 1.1 & 2 & 0.60 & 0.1668688700 & 0.088 & \(1\times10^{-14}\)\\
Thermal, \(\nu=0.1\) & 1.1 & 0 & 0.60 & 0.4917263119 & 0.088 & \(2\times10^{-15}\)\\
\end{tabular}
\end{ruledtabular}
\end{table*}

Table~\ref{tab:closure} shows that the fixed-point iteration
\(R^{(m+1)}=-(M^{(1)}+G^{(1)}R^{(m)})^{-1}F^{(1)}\) converges in one step
from \(R^{(0)}=0\) and from random starting matrices to the same
\(R_\infty\), that the resulting states agree with brute-force
diagonalization to \(10^{-13}\)--\(10^{-16}\), and
(Proposition~\ref{prop:terminal})
that replacing the terminal condition by \(X_{n_{\max}+1}=0\) changes
\(\langle\nop\rangle\) by less than \(2\times10^{-14}\) at \(N=3\)
(\(g=1.1\) and \(2.4\)), \(N=8\), and in the good-cavity regime.  At a
strong-coupling/high-pump point the low-radial coefficients differ from
an \(n_{\max}=10\) reference by \(6.4\times10^{-5}\), \(1.3\times10^{-7}\),
\(8.9\times10^{-11}\) for \(n_{\max}=4,6,8\).

\subsection{Grading, pivots, terminal independence, and consistency checks}
\label{app:branch-checks}

The structural statements of Sec.~\ref{sec:branch} are checked on the
assembled blocks (script \texttt{branch\_selection\_checks} of
Ref.~\cite{code}).  For \(N\le6\) at \(K=0\), for \(K=-1\), and for
\(\nu=0.1\): \(F^{(0)}=0\) to \(10^{-14}\); \(F^{(1)}\) and \(\mathcal N_K\)
have no matrix element that fails to lower the grade by one and \(G\)
none that fails to raise it; the nilpotency indices of \(F^{(1)}\), \(G\),
and \(\mathcal N_K\) equal \(N+1\) at \(K=0\) and the index of
\(R_\infty\) is \(\le N+1\) (e.g.\ 5 at \(N=6\)), all within the bound
\(N+|K|+1\).  The spectral norms entering
Proposition~\ref{prop:terminal} grow quickly with \(N\)
(\(\|R_\infty\|=0.9,2.1,7.5,35,1.9\times10^3\) for \(N=1,2,3,4,6\) at the
baseline parameters, \(s=0.6\)), so the proposition's threshold \(n_0\) is
far above the radial truncations used, yet the recursion forgets its
terminal matrix immediately: the spectral-norm difference between the
\(R_0\) obtained in double precision from \(R_\infty\) and from a random
terminal matrix of unit norm (or from Miller's \(R_{n_{\max}}=0\)) is
\(3\times10^{-16}\) at \(N=3\), \(s=0.9\), \(n_{\max}=30\), i.e.\
round-off, and \(2\times10^{-13}\) at \(N=6\), \(s=0.9\); at \(N=3\),
\(s=0.6\), \(n_{\max}=20\) the directly measured difference is
\(3\times10^{-34}\), far below round-off, meaning that the terminal
matrix is annihilated identically in floating point by the nilpotent
structure within the first few steps.  The super-exponential decay of
Theorem~\ref{thm:branch} is visible in the good-cavity regime, where the
state is broad: at \(N=2\), \(s=0.9\), \(n\,\|X_{n+3}\|/\|X_n\|\) falls
from 8.9 at \(n=2\) to 0.05 at \(n=48\).

\paragraph*{Finite pivots.}  Hypothesis (P) concerns the limiting
pivots and is not established by inverting the finite ones.  What the
computations monitor are the finite-cutoff pivots
\(P_j^{(n_{\max})}=M_j+GR_j^{(n_{\max})}\): their smallest singular values
and reciprocal condition numbers, and the diagnostic
\(\|(P_j^{(n_{\max})})^{-1}G\|\,\varepsilon_j\) with
\(\varepsilon_j=\|R_j^{(n_{\max})}-R_j^{(n_{\max}')}\|\)
(\(n_{\max}'=n_{\max}+20\) or \(+30\)) used as an \emph{estimate} of the
distance to the limiting transfer matrix.  This diagnostic is not an
upper bound --- two finite cutoffs can agree while sharing a common
error --- so a value below one supports, but does not establish, the
invertibility of the limiting pivot.  At \(N=3\) the reciprocal
condition numbers stay above \(3\times10^{-2}\) at every level and the
diagnostic is below one at every level (largest value \(0.8\), at the
terminal level, where \(\varepsilon_j\) is not small by construction);
in the good-cavity regime at \(N=3\), \(n_{\max}=90\), they stay above
\(9\times10^{-5}\) and the diagnostic is below \(10^{-9}\) for
\(j\le80\) (the per-level profiles are stored with the output of
Ref.~\cite{code}).  At \(N=8\), \(s=0.9\), the reciprocal condition
numbers fall from \(2\times10^{-4}\) at \(j=1\) to \(2\times10^{-10}\) at
the terminal level and the diagnostic is below \(2\times10^{-5}\) for
\(j\le30\) of \(n_{\max}=40\) but exceeds one near the terminal level;
at \(N=16\) it exceeds one at every level in Euclidean norm (reciprocal
condition numbers \(10^{-15}\)--\(10^{-29}\)), which does not show that
a limiting pivot is singular but marks the regime whose precision is
described in Sec.~\ref{sec:precision} and Table~\ref{tab:Nconditioning}.

\paragraph*{Photon-distribution and reduced-state consistency checks.}
The following checks are consistency checks on the computed solution,
not a positivity certificate for the joint operator: the photon
distribution reconstructed by binomial inversion of the factorial
moments, \(P(m)=\sum_{n\ge m}(-1)^{n-m}\binom nm c_n\), sums to one to
\(10^{-12}\), is non-negative to round-off (smallest value
\(-1.6\times10^{-24}\) at \(m=45\) in the good-cavity case,
\(+8\times10^{-74}\) in the baseline case), returns
\(\langle\nop\rangle=c_1\) to all digits, and the reduced atomic state
\(\rho_N\) is positive (smallest eigenvalue \(0.065\), \(0.11\), \(0.17\)
in the three cases).  For \(N=2,3\) we have also reconstructed the
\emph{joint} atom--field operator
\(\rho=\sum_{n\alpha}X_{n\alpha}\rho_n^{(k_\alpha)}\otimes R_{\mv_\alpha}\)
explicitly in a Fock space of dimension \(n_{\max}+N+2\): at the
baseline parameters (\(N=2,3\), \(s=0.6\); \(N=3\), \(s=0.9\)) it is
Hermitian to \(4\times10^{-16}\), has unit trace to \(10^{-15}\),
smallest eigenvalue above \(-10^{-33}\), stationarity residual
\(\|\mathcal L\rho\|_1\le9\times10^{-15}\) under the independent
brute-force Liouvillian, and trace distance \(\le2\times10^{-15}\) to
the brute-force stationary state; in the good-cavity case \(N=2\),
\(s=0.9\) (\(n_{\max}=50\)) the corresponding numbers are
\(3\times10^{-11}\), \(2\times10^{-12}\), \(-4\times10^{-15}\),
\(9\times10^{-10}\), and \(2\times10^{-10}\), consistent with the
\(e^{\mu}\) conditioning scale at \(\mu=6\).

\paragraph*{High-precision reference values.}  The good-cavity results
of Table~\ref{tab:precision} were recomputed with the recursion carried
out in 40-digit arithmetic (parameters as exact rationals, \texttt{mpmath}):
the reference values (identical to all digits shown for \(n_{\max}\ge90\); at \(N=5\) the radial truncation \(n_{\max}=60\) is itself not converged, differing by \(2\times10^{-9}\); at \(N=6\) only \(n_{\max}=120\) and \(160\) were run and agree) are \(N=2\): \(\langle\nop\rangle=6.09075065365421\), \(g^{(2)}(0)=1.03755772422335\); \(N=3\): \(\langle\nop\rangle=10.0312028676859\), \(g^{(2)}(0)=1.02116407001303\); \(N=4\): \(\langle\nop\rangle=14.0119349228883\), \(g^{(2)}(0)=1.01431233868912\); \(N=5\): \(\langle\nop\rangle=18.0015697314169\), \(g^{(2)}(0)=1.01082605807374\); \(N=6\): \(\langle\nop\rangle=21.9947770069797\), \(g^{(2)}(0)=1.00873208297469\).  The double-precision errors quoted in Sec.~\ref{sec:precision} are the distances of the Table~\ref{tab:precision} computations from these values.  The Poissonian rescaling of the radial coefficients changes the double-precision errors by factors between 0.5 and 2 only.

\subsection{Weak-pump series}
\label{app:series-checks}

The Taylor coefficients of Table~\ref{tab:series} are obtained from
the linear systems \(\mathcal M_0X^{(r)}=-\mathcal M_1X^{(r-1)}\) with
\(n_{\max}=N+3\) (script \texttt{weak\_pump\_series\_exact} of
Ref.~\cite{code}), which carry no radial truncation error but are solved
in floating point.  The block-Jacobi matrix is affine in \(s\) to
machine precision; \(\mathcal M_0\) has a one-dimensional kernel spanned
by \(\Phi_{0,(N,0,0,0)}\) (second-smallest singular value \(0.33\) at
\(N=3\), \(2\times10^{-5}\) at \(N=8\)); the trace-constrained solves have
residuals below \(5\times10^{-14}\); and the computed \(X^{(r)}_n\)
vanish for \(n>r\) to \(10^{-15}\) (\(N=3\)) and \(4\times10^{-10}\)
(\(N=8\)), as Lemma~\ref{lem:filtration} requires.  The full
continued-fraction solution at \(s=0.01,0.02,0.04\) agrees with the
series truncated at order \(N+2\) to the size of the first omitted term:
for \(N=5\), \(\Gamma_5\) and \(K_5\) agree to \(4\times10^{-8}\) and
\(7\times10^{-8}\) at \(s=0.01\) and to \(2\times10^{-5}\) and
\(8\times10^{-4}\) at \(s=0.04\); for \(N=6\), \(\Gamma_6\) and \(K_6\)
agree to \(4\times10^{-5}\) and \(3\times10^{-5}\) at \(s=0.02\)
(at \(s=0.01\), \(K_6\approx4\times10^{-15}\) is below the resolution of
the full solve in double precision, which is why the series rather than
finite-pump slopes is the appropriate test at high order).  Two further
checks address the floating-point accuracy of the coefficients
themselves (script \texttt{revision\_checks\_round2} of
Ref.~\cite{code}).  First, for \(N=6\) and \(N=8\) the systems were
represented exactly (interaction entries are \(ig/2\) times integers,
diagonal entries rationals) and solved in 40-digit arithmetic by
iterative refinement of the double-precision solution (residuals
decreasing from \(10^{-14}\)--\(10^{-8}\) to below \(5\times10^{-38}\) in
three refinements): the leading coefficients of the double-precision computation are confirmed to a relative \(5\times10^{-12}\) at \(N=6\) and \(1\times10^{-8}\) at \(N=8\) (median \(3\times10^{-10}\); the largest deviation is for \(\Gamma_8\)), and the coefficients below the proved orders, which are \(10^{-15}\)--\(4\times10^{-11}\) in double precision, are zero to \(3\times10^{-39}\) (\(N=6\)) and \(7\times10^{-37}\) (\(N=8\)) in 40-digit arithmetic, confirming that their double-precision values are round-off.  Second, for \(N=4,5\) the same
coefficients were computed by perturbation theory in the full
Fock\(\otimes\)atom space without any symmetry reduction (sparse LU of
the \(s=0\) Liouvillian with the trace constraint replacing one
equation, cavity cutoff \(N+4\)): the leading coefficients agree with those of Table~\ref{tab:series} to a relative \(2\times10^{-12}\) (\(N=4\)) and \(1\times10^{-12}\) (\(N=5\)), and the coefficients below the proved orders vanish to \(2\times10^{-15}\) or better.

\subsection{Spectral pole: convergence, selection, and line shape}

\begin{table*}[tbp]
\caption{Spectral checks.  \(\lambda_1\) at the largest radial truncation
(\(n_{\max}=40\) baseline, 80 good cavity); ``spread'' is the maximum change
of \(\lambda_1\) over \(n_{\max}\in\{16,24,30,40\}\) (baseline) or
\(\{40,60,80\}\) (good cavity); the selected pole is identical for residue
thresholds \(\theta=1,5,10,20,30\%\) in every row; \(S(\omega')>0\) on the
whole grid in every row.  The last two columns compare the full width at
half maximum of \(S(\omega')\) with \(-2\operatorname{Re}\lambda_1\).}
\label{tab:spectral-checks}
\begin{ruledtabular}
\begin{tabular}{lcccccc}
Regime & \(N\) & \(s\) & \(\lambda_1\) & spread & FWHM of \(S\) & \(-2\operatorname{Re}\lambda_1\)\\
\hline
Baseline (doublet) & 1 & 0.3 & \(-0.448708365\pm0.457879042i\) & \(4\times10^{-15}\) & 1.297 & 0.897\\
Baseline (doublet) & 2 & 0.3 & \(-0.479710666\pm0.651980285i\) & \(1\times10^{-14}\) & 1.781 & 0.959\\
Baseline (doublet) & 3 & 0.6 & \(-0.558125741\pm0.522869362i\) & \(8\times10^{-14}\) & 1.476 & 1.116\\
Baseline (merged) & 2 & 0.9 & \(-0.557682884\pm0.189185283i\) & \(2\times10^{-14}\) & 0.822 & 1.115\\
Baseline (merged) & 3 & 0.9 & \(-0.551741728\pm0.186902403i\) & \(2\times10^{-13}\) & 0.833 & 1.103\\
Good cavity & 1 & 0.3 & \(-0.117727607\) & \(6\times10^{-16}\) & 0.2278 & 0.2355\\
Good cavity & 2 & 0.3 & \(-0.192961969\) & \(6\times10^{-15}\) & 0.3514 & 0.3859\\
Good cavity & 3 & 0.6 & \(-0.051540779\) & \(7\times10^{-14}\) & 0.10266 & 0.10308\\
Good cavity & 2 & 0.9 & \(-0.012452464\) & \(3\times10^{-14}\) & 0.02488 & 0.02490\\
Good cavity & 3 & 0.9 & \(-0.009905695\) & \(2\times10^{-12}\) & 0.01980 & 0.01981\\
\end{tabular}
\end{ruledtabular}
\end{table*}

Table~\ref{tab:spectral-checks} supports the spectral statements of
Sec.~\ref{sec:spectrum} (the FWHM is the outermost half-height span of
\(S(\omega')\) on a grid of \(4\times10^5\) points, so for a resolved
doublet it spans both components): the dominant pole is converged in
\(n_{\max}\) to
\(2\times10^{-12}\) or better, its selection is independent of the residue
threshold between 1\% and 30\%, the spectral function is positive, and in
the single-line regime its FWHM equals \(-2\operatorname{Re}\lambda_1\) to
0.1--0.5\% at \(s\ge0.6\) (3--9\% at \(s=0.3\), where secondary poles still
carry residue).  In the doublet regime the pole width is the width of one
component and the FWHM of the merged line is larger; when the doublet
has just merged (baseline, \(s=0.9\)) the two overlapping components with
complex residues produce a single peak narrower than either pole.

\section{Maximum-entropy reconstruction}
\label{app:maxent}

\(\rho_3^{(2)}=Z^{-1}\exp(\sum_\alpha\lambda_\alpha F_\alpha)\) is
constrained by nine operators: the three one-body sums
\(S_a=\sum_j\sigma_a^{(j)}\) and the six symmetric two-body sums
\(\sum_{i<j}[\sigma_a^{(i)}\sigma_b^{(j)}+\sigma_b^{(i)}\sigma_a^{(j)}]\),
\(a\le b\), which fix the swap-symmetric pair marginal completely (its
Bloch tensor is symmetric, and the three antisymmetric combinations have
identically zero expectation).  The multipliers are found by Newton
iteration on the strictly convex dual potential
\(\Psi(\lambda)=\ln\operatorname{Tr}e^{\sum\lambda_\alpha F_\alpha}-\sum\lambda_\alpha\operatorname{Tr}(\rho_3F_\alpha)\),
whose Hessian is the Kubo--Mori covariance
\(\sum_{ij}(F_\alpha)_{ij}(F_\beta)_{ji}w_{ij}-\langle F_\alpha\rangle\langle F_\beta\rangle\),
\(w_{ii}=p_i\), \(w_{ij}=(p_i-p_j)/(\ln p_i-\ln p_j)\), in the eigenbasis of
\(\rho_3^{(2)}\); it converges to residuals below \(10^{-12}\) in five to
fifteen steps.

\begin{table}[tbp]
\caption{Irreducible three-body information and connected coherence,
exact versus two-body maximum entropy, at the baseline parameters.}
\label{tab:maxent}
\begin{ruledtabular}
\begin{tabular}{cccc}
\(s\) & \(I_3^{\rm irr}\) (nats) & \(\Gamma_3\) exact & \(\Gamma_3\) max-ent\\
\hline
0.3 & \(1.0757\times10^{-4}\) & \(2.4907\times10^{-3}\) & \(6.2794\times10^{-4}\)\\
0.5 & \(2.8916\times10^{-4}\) & \(4.4786\times10^{-3}\) & \(3.5722\times10^{-4}\)\\
0.7 & \(3.4826\times10^{-4}\) & \(5.0932\times10^{-3}\) & \(4.1042\times10^{-6}\)\\
0.9 & \(2.8217\times10^{-4}\) & \(4.7667\times10^{-3}\) & \(-1.5054\times10^{-5}\)\\
\end{tabular}
\end{ruledtabular}
\end{table}

At small pump \(I_3^{\rm irr}/s^2\) approaches \(5.5\times10^{-4}\) for
\(s\lesssim3\times10^{-3}\) and rises by 50\% by \(s=0.1\); we do not
promote this to an asymptotic law, since \(\rho_3(s{=}0)\) lies on the
boundary of state space and a rigorous entropy expansion would have to
control the resulting nonanalytic terms \cite{Zhou2008}.

\bibliographystyle{apsrev4-2}
\bibliography{refs}

\end{document}